\documentclass[journal]{IEEEtran}
\usepackage{amsthm, amsmath,amsfonts}
\usepackage{array}
\usepackage[caption=false,font=normalsize,labelfont=sf,textfont=sf]{subfig}
\usepackage{textcomp}
\usepackage{stfloats}
\usepackage{url}
\usepackage{verbatim}
\usepackage{graphicx}
\usepackage{cite}
\usepackage[linesnumbered, ruled, vlined]{algorithm2e}

\makeatletter
\newcommand{\removelatexerror}{\let\@latex@error\@gobble}
\makeatother
\usepackage{booktabs}
\usepackage{comment}
\usepackage{pgfpict2e}
\usepackage{float}
\usepackage{amsmath}
\newtheorem{theorem}{Theorem}

\usepackage{multirow}
\usepackage{supertabular}
\newcolumntype{R}{>{\raggedleft\arraybackslash}X}
\newcolumntype{C}{>{\centering\arraybackslash}X}
\newtheorem{proposition}{Proposition}

\usepackage{orcidlink}

\let\mycomment\comment
\let\comment\undefined
\let\originalwidth\textwidth

\let\comment\mycomment
\let\textwidth\originalwidth

\usepackage[switch]{lineno}

\begin{document}
	
	\title{\textsc{HyColor}: An Efficient Heuristic Algorithm for Graph Coloring}
	
	\author{Enqiang Zhu$^*$, Yu Zhang$^*$, Haopeng Sun, Ziqi Wei, Witold Pedrycz, \emph{Life Fellow, IEEE},\\ Chanjuan Liu$^\dag$, \emph{Member, IEEE} and Jin Xu$^\dag$
        \thanks{$*$ Enqiang Zhu and Yu Zhang contributed equally to this paper.

        $\dag$ Chanjuan Liu and Jin Xu are the corresponding authors.
        }
  
		\thanks{This work was supported in part by the National Natural Science Foundation of China under Grants (62272115), in part by Science and Technology Projects in Guangzhou.
		
		Enqiang Zhu and Haopeng Sun are with the Institute of Computing Science and Technology, Guangzhou University, Guangzhou 510006, China (e-mail: zhuenqiang@gzhu.edu.cn, 2112006193@e.gzhu.edu.cn). 

        Yu Zhang and Jin Xu are with the Key Laboratory of High Confidence Software Technologies, Ministry of Education, also with the School of Computer Science, Peking University, Beijing 100871, China (e-mail: yuzhang.cs@stu.pku.edu.cn; jinxu.pku@gmail.com)

        Ziqi Wei is with the CAS Key Laboratory of Molecular Imaging, Institute of Automation, Beijing 100190, China. (ziqi.wei@ia.ac.cn)
			
        Witold Pedrycz is with the Department of Electrical and Computer Engineering, University of Alberta, Edmonton, AB, Canada (e-mail: wpedrycz@ualberta.ca).
    
        Chanjuan Liu is with the School of Computer Science and Technology, Dalian University of Technology, Dalian 116024, China (e-mail: chanjuan.cs.ai@gmail.com).
        }
	}

	\maketitle
\begin{abstract}
    The graph coloring problem (GCP) is a classic combinatorial optimization problem that aims to find the minimum number of colors assigned to vertices of a graph such that no two adjacent vertices receive the same color. GCP has been extensively studied by researchers from various fields, including mathematics, computer science, and biological science. Due to the $\mathcal{NP}$-hard nature, many heuristic algorithms have been proposed to solve GCP. However, existing GCP algorithms focus on either small hard graphs or large-scale sparse graphs (with up to 10$^7$ vertices). This paper presents an efficient hybrid heuristic algorithm for GCP, named \textsc{HyColor}, which excels in handling large-scale sparse graphs while {\color{black}achieving} impressive results on small {\color{black}{dense}} graphs. The efficiency of \textsc{HyColor} comes from the following three aspects: a local decision strategy to improve the lower bound on the chromatic number; a graph-reduction strategy to reduce the working graph; and a {\color{black}$k$-core} and mixed degree-based greedy heuristic for efficiently coloring graphs. 
    {\color{black} \textsc{HyColor} is evaluated against three state-of-the-art GCP algorithms across four benchmarks, comprising three large-scale sparse graph benchmarks and one small dense graph benchmark, totaling 209 instances. The results demonstrate that \textsc{HyColor} consistently outperforms existing heuristic algorithms in both solution accuracy and computational efficiency for the majority of instances. Notably, \textsc{HyColor} achieved the best solutions in 194 instances (over 93\%), with 34 of these solutions significantly surpassing those of other algorithms. Furthermore, \textsc{HyColor} successfully determined the chromatic number and achieved optimal coloring in 128 instances.
    }
	\end{abstract}
	
	
\section{Introduction}\label{sec1}

With many important applications in the real world, combinatorial optimization problems~\cite{Hou23} have attracted great attention from both academia and industry~\cite{TianSZTJ23}. One of the most important challenges in combinatorial optimization is graph coloring, which is widely used to facilitate conflict resolution~\cite{WOS:001068977000001} or optimal demarcation of mutually exclusive events. A $k$-coloring of a graph $G$ is an assignment of $k$ colors to the vertices of $G$ such that every pair of adjacent vertices receives different colors.  The graph coloring problem (GCP) requires to determine the minimum $k$ (the chromatic number) for which $G$ admits a $k$-coloring \cite{brelaz1979new}. 

GCP plays a key role in both theory and practice. In theory, GCP is a classic $\mathcal{NP}$-hard problem. There are many famous mathematical conjectures on graph coloring, such as the four-color conjecture~\cite{four-color}, Haj\'{o}s' graph-coloring conjecture~\cite{hajos}, and Kneser's conjecture~\cite{greene2002new}. In practice, GCP can be used to model many real-world problems in fields such as operations research~\cite{gondran2019optimality}, computational biology~\cite{xue2022graph}, power systems~\cite{WOS:000631202400013}, communication networks~\cite{sadavare2012review}, machine learning~\cite{gupta2022scalable}, scheduling problem~\cite{9825671}, as well as routing and wavelength assignment~\cite{zhu2020partition}. As an example, we consider the simplest scheduling problem, which seeks to determine the minimum number of time slots to implement $n$ jobs, where two different jobs cannot use the same time slot if they share common resources (e.g., the same classroom for exams). This problem can be modeled as a graph $G$ whose vertex set represents the set of jobs, and two jobs are adjacent if they share common resources. Then, the chromatic number of $G$ is the minimum number of time slots, since vertices (jobs) with the same color can be implemented in the same time slot~\cite{garey1979}. This problem is equivalent to the {\color{black}resource-constrained} project scheduling problem as well~\cite{almeida2016priority}. Given the significance of GCP in theory and practice, it is crucial to develop efficient algorithms for solving GCP.

GCP is shown to be $\mathcal{NP}$-hard when $k\geq 3$~\cite{garey1979}, and cannot be approximated within $n^{1-\epsilon}$ unless $\mathcal{P}$=$\mathcal{NP}$~\cite{zuck2007}. Observe that determining whether a graph on $n$ vertices has a $k$-coloring can be trivially checked in time $O(k^n)$, and is too complex to complete when $n$ is very large. Many exact exponential-time algorithms for GCP have been proposed in the literature to improve this trivial upper bound. However, it is difficult to design an exact algorithm for GCP with running time $O((2-\epsilon)^n)$ for some $\epsilon>0$, but for small $k$, some better algorithms have been proposed. For $k=3$ and $4$, based on the branch-and-reduce paradigm and measure-and-conquer analysis technique, exact GCP algorithms can achieve a running time of $O(1.3289^n)$~\cite{3-coloring} and $O(1.7272^n)$~\cite{4-coloring}, respectively; counting the number of all $k$-colorings  can be solved in $O(1.588^n)$~\cite{count3} and $O(1.9464^n)$~\cite{4-coloring}, respectively.

Over the past decade, significant advances in computational methods and data science have produced several large-scale datasets, many of which can be abstracted into graphs (also called complex networks)~\cite{newman2003structure}. Although many real-world graphs are sparse and their structures may have some nice properties (e.g.,  the degree of vertices follows a power-law distribution), it is still hard to approximately solve GCP on these graphs~\cite{shen2012new}. Previous studies on GCP generally focused on small graphs with hundreds or thousands of vertices, or the computation of the optimal chromatic number~\cite{porumbel2013informed,furini2017improved,cseker2021exact}; most of these algorithms have high time and space complexity, and cannot obtain satisfactory solutions within a reasonable time on very large graphs. In contrast, a relatively small body of literature addresses GCP on large-scale graphs.

{\color{black}
In this paper, we seek to enhance the state-of-the-art performance in addressing  GCP on general graphs. We introduce a hybrid heuristic algorithm for GCP, named \textsc{HyColor}, which comprises three key phases:
(1) The first phase employs a local decision strategy to refine the lower bound of the chromatic number by identifying a clique. (2) The second phase incorporates a graph reduction strategy specifically designed for large-scale sparse graphs to minimize their size. (3) The final phase utilizes core decomposition alongside mixed degree-based greedy heuristics to efficiently color the graphs.
}
\textsc{HyColor} is evaluated across four established benchmarks, {\color{black} encompassing a total of 209 instances.} In comparison to three leading heuristic algorithms for addressing GCP, \textsc{HyColor} consistently produces a $k$-coloring with a smaller value of $k$, demonstrating its effectiveness.

The rest of this paper is organized as follows. Section~\ref{relatedwork} briefly reviews related work. Section~\ref{sec2}  provides some preliminaries. Section~\ref{sec3} describes the proposed \textsc{HyColor} algorithm. Section~\ref{sec-experiment} presents the experimental design and analyzes the results. Finally, in Section~\ref{sec-conclusion}, conclusions are offered, and future studies are suggested.

\section{Related Work} \label{relatedwork}

We review the development of heuristic GCP approaches from tabu-based algorithms,  greedy-based heuristic algorithms, and metaheuristic-based algorithms.

\begin{table*}[bthp] 
\centering
\caption{Comparison of algorithms for GCP}
\resizebox{\linewidth}{!}{
\begin{tabular}{l|l|m{7.5cm}|m{5cm}}
\hline
\textbf{Type} & \textbf{Algorithm} & \textbf{Advantages} & \textbf{Limitations} \\ \hline
\multirow{2}{*}{Metaheuristic-based} & CQEA~\cite{xu2022cuckoo} & Improves search using quantum evolutionary algorithm   & High computational resource demands limit scalability \\ \cline{2-4}
 & GA~\cite{marappan2022new} & Reduces complexity compared to traditional crossover operators & Limited to small instances; \newline multi-generation evolution requires time \\ \hline
\multirow{2}{*}{Tabu-based Heuristic} & $\textsc{Tlbo-Color}$~\cite{dokeroglu2021memetic} & Performs well on some challenging
large graphs & Requires parameter tuning; long runtime \\ \cline{2-4}
 & $\textsc{GC-SLIM}$~\cite{schidler2023sat} & Strong performance on dense graphs & High memory consumption \\ \hline
 Tabu\&Greedy-based &  $\textsc{LS+I-DSatur}$~\cite{hebrard2019hybrid} & Performs well on large sparse graphs & High dependency on parameter tuning \\ \hline
\multirow{4}{*}{Greedy-based} & $\textsc{Recolor}$~\cite{rossi2014coloring} & Fast and scalable for large networks & No reduction rules \\ \cline{2-4}
 & Verma's~\cite{verma2015solving} & Effective on sparse graphs using reduction rules & Limited performance on challenging graphs \\ \cline{2-4}
 & $\textsc{FastColor}$~\cite{lin2017reduction} & High efficiency on large sparse graphs & Reduced graph still remains quite large \\ \cline{2-4}
 & \textbf{HyColor (Proposed)} & Excels in large sparse graphs; competitive on dense graphs & Relies on vertex ordering \\ \hline
\end{tabular}
}
\label{tab:tax}
\end{table*}

\subsection{Tabu-based Heuristic Algorithms}
A typical tabu search algorithm for solving GCP, $\textsc{Tabucol}$, was proposed in 1987 by Hertz and de Werra~\cite{hertz1987using}. 
The idea of the early algorithms is to start with a $k$-coloring that contains conflict edges (i.e., edges whose endpoints receive the same color) and eliminate conflict edges through a sequence of color swaps, where the tabu is reflected in the color $c$ for a vertex $v$.  If $c$ is replaced by $c'$, then it cannot be used to color $v$ in the subsequent iterations~\cite{avanthay2003variable}.  
Bl\"{o}hliger and Zufferey (2008)~\cite{blochliger2008graph} proposed \textsc{Foo-Partialcol}, which constructs a coloring of a graph $G$ by extending partial solutions (i.e., a coloring of a subgraph of $G$ without conflict edges), and introduced a reactive tabu tenure to enhance the performance of the proposed heuristic and the $\textsc{Tabucol}$ heuristic. However, the reactive forbidden search requires adjustment of the forbidden list, which is time-consuming. Therefore, \textsc{Foo-Partialcol} cannot efficiently deal with large graphs.  
Wu and Hao (2012)~\cite{wu2012coloring} proposed the \textsc{Extracol} algorithm to color large graphs, using an adaptive tabu search strategy to find many pairwise disjoint independent sets, and a memetic algorithm~\cite{lu2010memetic} to color the residual graph by deleting the found independent sets. Although  \textsc{Extracol} achieved better experimental results than previous algorithms on 11 large graphs (with 1000 to 4000 vertices), it required five hours to find a better solution, which is inefficient and unacceptable for industrial applications. 
In 2018, Marappan et al.~\cite{marappan2018solution} and Moalic et al.~\cite{moalic2018variations} introduced a hybrid heuristic, combining a genetic algorithm and a tabu search to solve GCP, using the genetic algorithm to find better solutions and the tabu search to accelerate convergence. However, the algorithm requires many computations to implement the genetic process, and the performance depends on the parameter settings, which affects its application to large graphs. 
Dokeroglu et al. (2021)~\cite{dokeroglu2021memetic} proposed \textsc{Tlbo-color}, a robust version of the modulo teaching optimization algorithm for solving GCP, incorporating a robust forbidden search and parallel metaheuristics. \textsc{Tlbo-Color} was shown to perform well on some challenging large graphs and ran faster than previous algorithms. However, \textsc{Tlbo-Color} often requires nearly 30 minutes to obtain a better solution, and many parameters must be tuned. The above hybrid heuristic algorithms are not guaranteed to capture a promising region when restarting the local search from a new point created by the crossover operator. To overcome this issue, Goudet et al (2022).~\cite{goudet2022deep} used a deep learning neural network to extract feature information of the structures of graphs, based on which a tabu search strategy was proposed to search for a better solution. However, the algorithm converges slowly, and its performance is limited by hardware since significant computing resources are needed to train the deep learning model. 
{\color{black}Recently, Schidler and Szeider \cite{schidler2023sat} proposed a hybrid approach known as \textsc{GC-SLIM} for addressing GCP. This method integrates tabu search with the SAT-based local improvement method (SLIM)~\cite{lodha2019sat,reichl2023circuit}. \textsc{GC-SLIM} iteratively refines the current coloring by selecting small subgraphs and optimally coloring them using SAT solvers, such as Glucose~\cite{audemard2009predicting} and Cadical~\cite{fleury2020cadical}. The method demonstrates strong performance on dense graphs, consistently producing high-quality coloring solutions. However, it faces considerable challenges when managing large-scale graphs, primarily due to the necessity of constructing an adjacency matrix for the graph, particularly in environments with limited memory.
}

\subsection{Greedy-based Heuristic Algorithms}
{\color{black}
In practice, greedy-based heuristic algorithms are commonly employed {\color{black}to solve} challenging problems~\cite{zhao2024iterative,7938718}.}
In the common framework of greedy heuristics for GCP, the vertices of a graph $G$ are arranged in linear order, say $x_1,x_2,\ldots,x_n$, and are colored, one by one, in this order with the smallest available color, i.e., each $x_i$ is colored with the smallest color $c$ that has not yet been assigned to a vertex in $N_{G}(x_i)\cap \{x_1,x_2,\ldots, x_{i-1}\}$. The quality of a coloring obtained by a greedy coloring heuristic is generally poor, and this depends heavily on the predefined vertex order. Note that there are in total $n!$ distinct vertex orders for an $n$-vertex graph, and it is hard to find an ordering based on which optimal coloring can be produced. Therefore, the greedy coloring heuristic is often combined with other techniques to improve the quality of solutions.

Rossi et al. (2014)~\cite{rossi2014coloring} proposed a local search algorithm, \textsc{Recolor}, to improve the quality of solutions obtained by greedy colorings based on different vertex-ordering methods. The algorithm does not perform well on hard graphs and applies no reduction rules, affecting its performance on large graphs. Verma et al. (2015)~\cite{verma2015solving} introduced reduction rules for GCP based on $k$-core and $k$-community and used various greedy heuristics to find lower and upper bounds on the chromatic number. Although this algorithm performed well on simple sparse graphs, it struggled with hard {\color{black}graphs}. Lin et al. (2017)~\cite{lin2017reduction} proposed a reduction rule based on a lower bound on the chromatic number and developed the \textsc{FastColor} algorithm with an iterative greedy coloring strategy to color large-scale sparse graphs. 
{\color{black}
Hebrard et al. (2019)~\cite{hebrard2019hybrid} introduced a hybrid coloring algorithm named \textsc{LS+I-DSatur}. This algorithm computes upper and lower bounds by exploiting graph degeneracy and cliques, respectively. It begins by simplifying the graph through various reduction techniques. Subsequently, it enhances the upper bound using \textsc{DSatur} in conjunction with Tabu search methods, ultimately performing the exact coloring process with iterated \textsc{DSatur} (\textsc{I-DSatur}). Since this algorithm uses tabu search, it can also be classified as a Tabu-based heuristic.}

\subsection{Metaheuristic-based Algorithm}
{\color{black}
Metaheuristic-based algorithms are commonly utilized to tackle various challenging problems, such as scheduling~\cite{pan2020knowledge, zhao2022hyperheuristic}, dominating sets~\cite{zhu2024dual}, and also GCP~\cite{7835722}.}  
In 2020, Mostafaie et al.~\cite{mostafaie2020systematic} systematically reviewed metaheuristic approaches for GCP. 
In the following, we present some recent advancements related to this topic.

Meraihi et al. (2018)~\cite{meraihi2019chaotic} proposed a salp swarm algorithm (SSA) for GCP,  using a logistic map instead of random variables.
While the SSA performs well compared with previous metaheuristic algorithms for GCP, it is still limited to small hard instances with at most 200 vertices. Silva et al. (2020)~\cite{silva2020improved} proposed an ant colony optimization (ACO) algorithm for GCP, with some heuristic strategies to improve the quality of solutions. However, the performance of the ACO algorithm {\color{black}depends} significantly on its parameter settings and the quality of the initial solution.
The particle swarm optimization method (PSO) is often used to solve some operations research problems~\cite{wang2012hybrid}.
Marappan et al. (2021)~\cite{marappan2021solving} devised a partition-based turbulent particle swarm optimization (DCTPSO) model to overcome stagnation when solving GCP. However, the algorithm requires the tuning of multiple parameters, and partitioning a graph may take a large amount of time. 
Xu et al. (2022) \cite{xu2022cuckoo} combined a quantum-inspired evolutionary framework with a cuckoo search strategy in the proposed CQEA algorithm for solving GCP. Although the quantum evolutionary algorithm can improve the search performance, the quantum gate operation requires many computational resources, which limits its generalization to large-scale graphs.  
Marappan et al.~\cite{marappan2022new} proposed a genetic algorithm (GA) to solve GCP, with a new stochastic operator to accelerate convergence.
While the algorithm reduces complexity compared to traditional crossover operators, the running target of the GA is still restricted to small hard instances, since the multi-generational evolution of large-scale graphs may require longer computation times. 
ACO and GA are the most popular metaheuristic algorithms for solving GCP. However, compared to hybrid heuristics based on reduction and greedy strategies, metaheuristics have relatively weak efficiency and accuracy, especially for large-scale networks. 

{\color{black}
Among the numerous GCP algorithms available, we compare a selection of advanced algorithms to elucidate their respective advantages and limitations, as outlined in Table~\ref{tab:tax}. Currently, \textsc{FastColor} and \textsc{LS+I-DSatur} are recognized as the leading algorithms for coloring large-scale sparse graphs, while \textsc{GC-SLIM} excels with small dense graphs and those of medium size. Each of these algorithms offers unique strengths in graph coloring; however, achieving a balance in coloring quality between small dense graphs and large-scale sparse graphs continues to pose a challenge. To tackle this issue, this paper introduces \textsc{HyColor}, designed to attain effective coloring quality for both large sparse graphs and small dense graphs through thoughtfully well-designed strategies.
}

\section{Preliminary Work} \label{sec2}

All graphs considered in this paper are simple finite undirected graphs.  Given a graph $G$, an edge {\color{black}$e=uv$ is said to be \emph{incident} with its two endpoints $u$ and $v$}, and \emph{vice versa}. If two distinct vertices are incident with the same edge, then they are \emph{adjacent}, and one is said to be a \emph{neighbor} of the other. The \emph{degree} $d_G(u)$ of a vertex $u\in V(G)$ is the number of edges incident with $u$; the \emph{neighborhood} of $u$, denoted by $N_G(u)$, is the set of vertices adjacent to $u$, and  $N_G[u]=N_G(u)\cup \{u\}$. Clearly, $d_G(u)= |N_G(u)|$. For any $V'\subseteq V(G)$,  the graph obtained from $G$ by deleting the vertices in $V'$ and their incident edges is denoted by $G-V'$, and  $G[V']=G-(V(G)\setminus V')$ is called the subgraph of $G$ \emph{induced} by $V'$.

A \emph{proper coloring} of a graph $G$ is an assignment of colors to $V(G)$ such that any two adjacent vertices receive different colors. Given a proper coloring $f$ of $G$, we use $CS(f)$ to denote the set of colors used by $f$, and denote by $f|_H$ the restriction of $f$ to a subgraph $H$ of $G$. If $|CS(f)|=k$, then we call $f$ a $k$-coloring of $G$. The minimum $k$ such that $G$ admits a $k$-coloring is called the \emph{chromatic number} of $G$, denoted by $\chi(G)$. In this paper, we do not consider other types of colorings{\color{black}; when we refer to a coloring, we exclusively mean} a proper coloring. A $k$-coloring is called an \emph{optimal coloring} of $G$ if $k=\chi(G)$.

Two important notions that are closely related to coloring are independent sets and cliques. For a graph $G$, an \emph{independent set} (resp. a clique) of $G$ is a subset $S \subseteq V(G)$ such that $G[S]$ is an empty graph (resp. a complete graph).  A $k$-independent partition of $G$ is a set $\{V_1,V_2,\ldots, V_k\}$ of cardinality $k$ such that $V_i$ is an independent set, $V_i\cap V_j=\emptyset$, and $V(G)=V_1\cup V_2 \cup \ldots \cup V_k$, where $i,j \in \{1,2,\ldots,k\}$ and $i\neq j$.
Denote by $\alpha(G)$ the smallest integer $k$ for which $G$ admits a $k$-independent partition. Note that every $k$-independent partition of $V(G)$ corresponds to a $k$-coloring $f$ of $G$ such that $f(v)=i$ for every $v\in V_i, i=1,2,\ldots,k$. Therefore, $\chi(G)=\alpha(G)$.
A clique $S$ of a graph $G$ is \emph{maximum} if no clique $S'$ of $G$ satisfies $|S'|>|S|$. Denote by $\omega(G)$ the number of vertices of a maximum clique of $G$. Clearly, $\chi(G)\geq \omega(G)$.
Let $H$ be a subgraph of $G$. Then, $\chi(G)\geq \chi(H)$, since $f|_H$ is a coloring of $H$ for any coloring $f$ of $G$.

From the above discussion, we have the Proposition~\ref{pro-1}.
\begin{proposition}\label{pro-1}
	Let $G$ be a graph with a $k$-independent partition, $H$ a subgraph of $G$, and $C_H$  a clique of $H$. Then
	$k\geq {\color{black}\alpha(G) = \chi(G)}\geq \chi(H)\geq \omega(H)\geq |C_H|.$	
\end{proposition}

Proposition \ref{pro-1} provides a method of calculating a lower and upper bound for $\chi(G)$. Although $\alpha(G)$ and $\omega(H)$ can give better upper and lower bounds, determining them exactly in a graph is $\mathcal{NP}$-hard \cite{garey1979}. Therefore, in practice, one can resort to heuristic approaches to find a clique with cardinality as large as possible instead of $\omega(G)$ and a $k$-independent partition with $k$ as small as possible instead of $\alpha(G)$.

\section{H{\scriptsize Y}C{\scriptsize OLOR} Algorithm}\label{sec3}

 We introduce the \textsc{HyColor} algorithm, which receives a graph as input and outputs a coloring and the number of colors the coloring uses. 
 We define the following notations. 
 {\color{black}Let \( G \) and \( G_k \) represent the original graph and the working graph, respectively. The symbols \( lb^* \) and \( ub^* \) denote the best-known lower and upper bounds on the chromatic number of \( G \). The variable \( last\_lb_e \) reflects the lower bound on the chromatic number of \( G \) from the previous iteration, while \( f^* \) represents the current best-found coloring. The variable \( del\_stack \) is a stack that maintains a chronological order of removed vertices, enabling \( f^* \) to extend into a coloring of \( G \). Furthermore, the Boolean variable \( exactlb \) indicates whether the \textsc{Exactlb}() function (refer to Section~\ref{sec-exactlb}) is being employed for clique computation, while the Boolean variable \( reduced \) specifies whether \( G_k \) has undergone reduction.}

As shown in Algorithm \ref{alo-1}, \textsc{HyColor} first initializes the working graph $G_k$ as $G$, $lb^*$ {\color{black}and $last\_lb_e$} as 0, $ub^*$ as the order $|V|$ of $G$, $f^*$ as an empty function, $del\_stack$ as an empty stack, {\color{black}$exactlb$ as \emph{true},} and $reduced$ as \emph{false}. 

The algorithm then enters into a loop in which a coloring of $G_k$ is iteratively improved until an optimal coloring is proved, i.e., $lb^*=ub^*$ {\color{black}(lines 26 and 27), or until all vertices are removed according to graph reduction (lines 18 and 19), or a time limit $cut\_off$ is reached (lines 3--27).} Finally, \textsc{HyColor} returns the best found coloring $f^*$ of $G$. 
{\color{black}Each iteration of \textsc{HyColor} is primarily composed of three main components.}

\begin{algorithm}[htbp]
	\caption{$\textsc{HyColor}(G)$} \label{alo-1}
	\KwIn{An undirected graph $G = (V,E)$}
	\KwOut{A coloring $f^*$ of $G$ and $|CS(f^*)|$}
	${G_k}\leftarrow G$, $lb^*\leftarrow 0,u{b^*}\leftarrow |V|$, ${f^*}\leftarrow \emptyset, last\_lb_e\leftarrow 0$\;
    $del\_stack\leftarrow \emptyset$, $exactlb\leftarrow true$, $reduced \leftarrow false$\;
    
	\While{$elapsed\_time < cut\_off$}{
		$last\_lb \leftarrow lb^*$\;
		$lb^* \leftarrow$ max($lb^*$, $\textsc{Findclq}(G_k)$)\;

        \If{$lb^* > last\_lb$}{
            $(G_k, S, reduced)\leftarrow \textsc{ReduRule} (G_k, lb^*)$\;
            $del\_stack \leftarrow del\_stack\cup S$\;
        }
		
		\If{$exactlb$ is true}{
  
			$lb_{e} \leftarrow$ $\textsc{Exactlb}(G_k, lb^*)$\;
            \If{$lb_{e} > lb^*$}{
                $lb^*\leftarrow lb_{e}$\;
                $(G_k, S, reduced)\leftarrow \textsc{ReduRule} (G_k, lb^*)$\;
                $del\_stack \leftarrow del\_stack\cup S$\;
            }
            \If{$lb_{e} = last\_lb_e$}{
                $exactlb \leftarrow false$\;
            }
            $last_{e}\leftarrow lb_{e}$\;
		}

        \If {$V(G_k)$ = $\emptyset$ }{
            \textbf{break};
        }

		\If {$reduced $ \emph{or} $ub^*=|V|$}{
			
			$f \leftarrow $ \textsc{MddColor}($G_k$, $f^*|_{G_k}, mdd\_set$); // $f^*|_{G_k}$ is the restriction of $f^*$ to $G_k$
		}
		\Else{
			$f \leftarrow $ \textsc{Dsatur}($G_k$);
		}

		\If {$|CS(f)| < u{b^*}$}{
			${f ^*} \leftarrow f,u{b^*} \leftarrow |CS(f)|$;
		}
		
		\If {$u{b^*} = l{b^*}$}{
			
			\textbf{break};
			
		}

	}
	
	$f^* \leftarrow $ a coloring of $G$ extended from $f^*$ by coloring vertices in $del\_stack$ according to\textbf{BF-Rule};
	
	\Return{$f^*$ \emph{and} $|CS(f^*)|$}
	
\end{algorithm}

The first {\color{black}component} serves for the computation of a high-quality lower bound on $\chi(G)$ {\color{black}(lines 4, 5, 10--12)} {\color{black}for assisting graph reduction}, using an efficient heuristic finding clique algorithm \textsc{Findclq}$()$ to find {\color{black}a clique $C$} of $G_k$ (simplified from FindClique$()$~\cite{cai2016fast}, proposed by Cai and Lin, 2016). Note that $\chi(G)\geq \chi(G_k)\geq |C|$. 
{\color{black}The algorithm employs the clique \( C \) that is found at the start of the process, specifically the one with the largest size \( |C| \), as the lower bound \( lb^* \). If a superior lower bound is identified in any iteration, the algorithm swiftly utilizes this new bound to simplify the instance, as outlined in lines 6–8. 
} 
The value of $lb^*$ is refined through the application of a novel heuristic approach, designated as \textsc{Exactlb}() (refer to Section~\ref{sec-exactlb}), {\color{black}under the condition that the Boolean variable $exactlb$ is set to true (lines 9--17).}


The second {\color{black}component} is a graph {\color{black}reduction} strategy, \textsc{ReduRule()} (see Section~\ref{sec-redu-1}), which is used to iteratively reduce the working graph {\color{black}(lines 6--8, 11--14). It effectively reduces the size of the instance and the problem scale, thereby increasing the algorithm's likelihood of approaching an optimal solution.} Each time a larger lower bound $lb^{\color{black}*}$ on $\chi(G_k)$ is obtained, the algorithm searches for vertices in $G_k$ with degree less than $lb^{\color{black}*}$. 
If such vertices exist, the working graph $G_k$ is reduced using the $\textsc{ReduRule}()$ procedure (see Section~\ref{sec-redu-1}), and the stack $del\_stack$ is updated by adding the new removed vertices one by one and {\color{black}$reduced$} is set to $true$. It should be noted that when the lower bound $lb^{\color{black}*}$ is not improved or there are no vertices with degree less than $lb^{\color{black}*}$, no vertex can be deleted by $\textsc{ReduRule}()$, and $G_k$ remains unchanged.

The third {\color{black}component is a coloring scheme, \textsc{MddColor()}, based on the concept of $k$-core and} mixed degrees, as explained in Section~\ref{sec-mddcolor}. 
{\color{black}The sorting based on core decomposition and mixed degree is relatively time-consuming, and the vertex sequence changes only when the working graph is reduced.}
Based on this observation, \textsc{MddColor()} is implemented only in the case that {\color{black}$reduced$ is $ture$ or the \textsc{HyColor} algorithm is just in its first round of iteration (lines 20 and 21)}. Moreover, if the reduction fails in the $i${\color{black}-}th ($i\geq 2$) iteration, then a classical coloring algorithm $\textsc{Dsatur()}$~\cite{brelaz1979new} is employed to color the working graph $G_k$ ({\color{black}lines 22 and 23}). 
If the above coloring procedures can obtain a better coloring $f$ than the current {\color{black}best-found} coloring $f^*$, then $f^*$ is replaced by $f$ ({\color{black}lines 24 and 25}). Note that $lb^*=ub^*$ implies that an optimal coloring has been obtained and the loop is terminated ({\color{black}lines 26 and 27}). Finally, the obtained coloring $f^*$ of $G_k$ is extended to a coloring of the original graph $G$ by coloring vertices in the stack $del\_stack$ that stores the removed vertices according to \textbf{BF-Rule} ({\color{black}line 28}), as explained in Proposition~\ref{pro-2}.

It should be noted that each time the algorithm enters a new round of iterations, it attempts to improve the current lower bound on the chromatic number, and then the working graph $G_k$ may be reduced based on the new lower bound. In addition, when $G_k$ is reduced, the upper bound of the chromatic number of $G_k$ can be improved with a certain probability by using \textsc{MddColor} or \textsc{Dsatur}. Therefore, in each iteration, the algorithm sequentially performs graph reduction and the search for a better upper bound.

We next describe the three important parts used in each iteration of $\textsc{HyColor}$, i.e., $\textsc{Exactlb()}$ for calculating the lower bound on the chromatic number, $\textsc{ReduRule()}$ for reducing the working graph, and $\textsc{MddColor()}$ for coloring the reduced graph.

\subsection{E{\scriptsize XACTLB} Algorithm} \label{sec-exactlb}

The \textsc{Exactlb()} algorithm is used to improve a lower bound on the chromatic number, i.e., the cardinality $lb^{\color{black}*}$ of a clique  obtained by {\color{black}the} heuristic algorithm $\textsc{Findclq}(G_k)$.
It should be noted that sometimes $lb^{\color{black}*}$ is far less than $\chi(G_k)$. This inspires us to design a heuristic approach to improve $lb^{\color{black}*}$.

The \textsc{Exactlb()} algorithm is based on Proposition \ref{pro-1}, which iteratively constructs small and dense subgraphs, and employs an {\color{black}advanced industrial solver to find larger cliques as better lower bounds on the chromatic number of $G_k$.
As shown in  Algorithm \ref{exactlb}, \textsc{Exactlb()} accepts a graph $G_k$ and a lower bound $lb_k$} on $\chi(G_k)$ as inputs, and outputs an improved lower bound.

First, a vertex $u$ with the maximum degree is selected and the test set $test\_set$ is initialized as $N_{G_k}[u]$, where the subgraph of $G_k$ induced by $test\_set$ is used to estimate $\chi(G_k)$ (line 1). After initialization, the algorithm enters a loop in which the size of $test\_set$ gradually increases.
{\color{black}
In each iteration, the algorithm starts with expanding the $test\_set$ by adding a vertex $v$ along with its neighbors, where  $v$ is selected as a neighbor of a vertex already in the $test\_set$ that has the maximum number of neighbors not currently included in the $test\_set$ (lines 3 and 4). Let $G_t$ denote the subgraph of $G_k$ induced by the $test\_set$. The advanced industrial solver Gurobi is then employed to compute a high-quality clique, which helps enhance the lower bound on the chromatic number $\chi(G_k)$.
It is crucial to note that Gurobi can only be activated when the size of the $test\_set$ is below a predefined threshold, referred to as $size\_upper$. This limitation is implemented to improve the efficiency of the iterations of the \textsc{HyColor} algorithm, given that Gurobi tends to be relatively time-consuming.
If the size of the resulting clique, denoted as $lb_{grb}$, surpasses the current lower bound $lb_k$, then the function \textsc{Exactlb()} will return $lb_{grb}$. If the time limit $cut\_off$ is reached or if the size of the $test\_set$ exceeds $size\_upper$, then $lb_k$ will be returned as the result.
Considering that common challenging small hard graphs have fewer than 1000 vertices and given Gurobi's solving capabilities, the value of $size\_upper$ is set to 1000. Additionally, since Gurobi iterates to find the optimal solution, which can be time-consuming, we aim for it to quickly return a feasible clique. Therefore, the $cut\_off$ for Algorithm~\ref{exactlb} is set to 1 second.
}

\begin{algorithm}[!t]
	\caption{\textsc{Exactlb}($G_k,lb_k$)}\label{exactlb}

	\KwIn{An undirected graph $G_k = (V_k,E_k)$, $l{b_k}$}
	\KwOut{An updated $l{b_k}$}
	
	$test\_Set \leftarrow N_{G_k}(u)$ // $u$ is an arbitrary vertex with the maximum degree in  $G_k$;
	
	\While{$elapsed\_time <  cut\_off$}{	
		
		$v \leftarrow $  pop a vertex from $V_k- test\_Set$ such that $N_{G_k}(v)\cap test\_Set \neq \emptyset$ and $N_{G_k}(v)-test\_Set$ is the maximum;
		
		$test\_Set \leftarrow text\_Set \cup  N_{G_k}[v]$;
		
		${G_t} \leftarrow G_k[text\_Set]$
  
			
		\If{$|test\_Set| < size\_upper$}{
		  $lb_{grb} \leftarrow$ $\textsc{Gurobi}(G_t)$\;
            \If{$lb_{grb} > lb_k$}{
			 \Return{$lb_{grb}$};
		  }
        }\Else{
            \Return{$l{b_k}$};
        }

	}
	
	\Return{$l{b_k}$};
	
\end{algorithm}


\subsection{Graph Reduction} \label{sec-redu-1}

The graph reduction \textsc{ReduRule}() is based on a lower bound on the chromatic number. As shown in Algorithm \ref{reduced-rule}, for a graph $G_k$ and a given lower bound $\ell$ on $\chi(G_k)$, the process of reducing $G_k$ is to iteratively delete a vertex with degree less than $\ell$; such a vertex is called an \emph{$\ell^-$-vertex}. Specifically,
if there exists an $\ell^-$-vertex in the working $G_k$, then $G_k$ is updated by deleting all $\ell^-$-vertices and adding them (one by one) to an ordered set $S$; we repeat this process until $G_k$ contains no $\ell^-$-vertex (lines 2--5). Compared with the reduction rule  \textbf{BIS-Rule} proposed by Lin et al. \cite{lin2017reduction}, the set of deleted vertices by \textsc{ReduRule}() is not necessarily an independent set, and \textsc{ReduRule}() emphasizes the dynamic degree when deleting vertices in one round of iteration in $\textsc{HyColor}()$, i.e., a deleted vertex $x$ cannot be an $\ell^-$-vertex of $G_k$ but it is an $\ell^-$-vertex of $G_k-S'$, where $S'$ is the set of vertices deleted before $x$. This can provide a better chance to delete more vertices, at least in one round of iteration, and accelerate the speed of reaching  good lower and upper bounds on the chromatic number.

\begin{algorithm}[!h]
	\caption{\textsc{ReduRule}($G_k, \ell$)}\label{reduced-rule}

	\KwIn{An undirected graph $G_k=(V,E)$, an integer $\ell$}
	\KwOut{An reduced graph and an ordered set $S$ of vertices}
	
	$S \leftarrow \emptyset$;
	
	\While{there exists a $\ell^-$-vertex}{
		
		$S'\leftarrow$ the set of $\ell^-$-vertices in  $G_k$;
		
		$S \leftarrow S \cup S'$;  
		
		$G_k \leftarrow G_k-S'$;
		
	}
	\Return{$G_k$ \emph{and} $S$};
	
\end{algorithm}

For a graph $G$ and a lower bound $\ell$ on $\chi(G)$, let $(G',S)=$\textsc{ReduRule}($G,\ell$) be the reduced graph and the set of deleted vertices returned by \textsc{ReduRule}($G,\ell$). The following proposition ensures that any coloring of $G'$ can be extended to a coloring of $G$ using at most $\chi(G)$ colors.

\begin{proposition}\label{pro-2}
	Given a graph $G$ and a lower bound $\ell$ on $\chi(G)$, let $(G', S)=\textsc{ReduRule}(G,\ell)$, and $f'$ a $k$-coloring of $G'$.  Then $\chi(G)\leq \max\{\ell, k\}$.
\end{proposition}
\textbf{Proof.}
Suppose that $S=\{v_1,v_2, \ldots, v_p\}$, where $v_i$ is deleted before $v_j$ for any $i<j, i,j\in\{1,2,\ldots, p\}$. Let $G_1=G$ and $G_i=G-\{v_1,\ldots, v_{i-1}\}$  for $i=2,3,\ldots, p$. Then $G'=G_p-\{v_p\}$, and $v_i$ is an $\ell^-$-vertex in $G_i$ for $i=1,2,\ldots,p$. Now, construct a coloring of $G_i$, denoted by $f_i$, by extending a coloring $f_{i+1}$ of $G_{i+1}$, for $i=1,2,\ldots, p$ in the following way, where $G_{p+1}=G'$.

First, let $f_{p+1}=f'$ and $C_p=\{f'(v)| v\in N_{G_p}(v_p)\}$. Note that $v_p$ is an $\ell^-$-vertex of $G_p$, which implies that $v_p$ has at most $\ell-1$ neighbors in $G_p$;
it follows that $|C_p|\leq \ell-1$, and thus $f_p$ can be constructed based on $f'$ by coloring $v_p$ with a color not appearing on neighbors of $v_p$ in $G_p$, i.e., let $f(v_p)$=$c_p \in \{1,2,\ldots, \ell\} \setminus C_p$ and $f_p(v)=f'(v)$ for any $v\in V(G_p)\setminus \{v_p\}$. Next, suppose that $f_i$ has been constructed for some $i\in \{2,\ldots, p\}$. Then $f_{i-1}$ is constructed in the same way, i.e., on the basis of $f_i$, we color vertex $v_{i-1}$ with an arbitrary color in $\{1,2,\ldots, \ell\}\setminus C_{i-1}$, where $C_{i-1}=\{f_{i}(v)| v\in N_{G_{i-1}}(v_i)\}$. Finally, $f_1$ is a coloring of $G$. Since, in the process of constructing $f_1$, vertices in $S$ use at most $\ell$ colors, it follows that  $\chi(G)\leq |CS(f_1)|\leq  \max\{\ell, k\}$. 
\qed

Proposition \ref{pro-2} guarantees that a coloring $f$ of a graph $G$ generated by extending a $k$-coloring $f'$ of a reduced subgraph $G'$ (obtained by \textsc{ReduRule()}) is either an optimal coloring of $G$ (when $k\leq \ell$ or $k=\chi(G)$) or a coloring using exactly $k$ colors (when $k>\chi(G)$). The rule of extending $f'$ to $f$ is denoted by \textbf{BF-Rule}.

To illustrate \textsc{ReduRule()}  and  \textbf{BF-Rule}, an example is presented in Figure \ref{fig-1}, which is a graph $G$ with vertex set  $V(G)=\{v_i|i=1,2,\ldots,13\}$. It is easy to see that the clique number of $G$ is 4, where $\{v_4,v_5,v_{12},v_{13}\}$ is a clique. Let $\ell=4$ be a lower bound on $\chi(G)$. In $G$, there are three $4^-$-vertices $v_1,v_3,v_6$; we delete $v_1,v_3,v_6$ from $G$ and let $G_1=G-\{v_1,v_3,v_6\}$. In $G_1$, there are two $4^-$-vertices $v_4,v_5$; we  delete $v_4,v_5$ from $G$ and let $G_2=G_1-\{v_4,v_5\}$. In $G_2$, there are no $4^-$-vertices; see the subgraph labeled with dashed lines in Figure \ref{fig-1}.  Now, let $f'$ be a 5-coloring of $G_2$, which is defined as  $f'(v_8)=f'(v_{12})=1,f'(v_9)=f'(v_2)=2,f'(v_{10})=f'(v_7)=3,f'(v_{11})=4,f'(v_{13})=5$. We then extend $f'$ to a 5-coloring $f$ of $G$ by coloring the deleted vertices in the order $v_5,v_4,v_6,v_3,v_1$. Clearly, $f(v_5)=2,f(v_4)=3,f(v_6)=4,f(v_3)=1$, and $f(v_1)=4$ is one extension scheme.

\begin{figure}[!h]
	\centering
	\includegraphics[width=3.5cm]{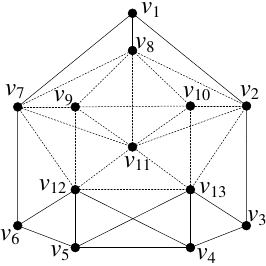}\\
	\caption{A graph $G$ on 13 vertices}\label{fig-1}
\end{figure}

\subsection{M{\scriptsize DD}C{\scriptsize OLOR} Algorithm} \label{sec-mddcolor}

Observe that using the greedy heuristic to color a graph depends heavily on a vertex ordering chosen for coloring the vertices individually. A suitable order can significantly reduce the number of colors the greedy coloring uses. Theoretically, there is always an ordering (called optimal ordering) based on which the greedy heuristic can construct an optimal coloring. However, it is difficult to determine an optimal ordering in advance. In this study, {\color{black} we establish a vertex ordering for coloring by integrating the concepts of $k$-core~\cite{seidman1983network} and mixed degree~\cite{zeng2013ranking}, which are delineated as follows.}

{\color{black}For a graph $G$, a $k$-core is a subgraph $H$ of $G$ such that every vertex in $H$ has degree at least $k$ and no vertex in $V(G)\setminus V(H)$ is adjacent to more than $k$ vertices in $H$. The maximum $k$ such that a vertex $v$ belongs to a $k$-core is called the \textit{shell number} of $v$. The \textit{core decomposition} of $G$ is to determine the shell numbers of all vertices of $G$, which can be achieved efficiently \cite{batagelj2003m}.}

For a vertex subset $S\subseteq V(G)$, the \emph{mixed degree of a vertex $v\in (V(G)\setminus S)$ associated with $S$} is defined as $d^m_{S}(v) = {d^r_{S}}({v}) + \lambda \times d^e_{S}(v)$,
where ${d^r_S}({v})$ and $d^e_{S}(v)$ denote the number of neighbors of $v$ in $V\setminus S$ and $S$, respectively, and weight coefficient $\lambda$ is a tunable real number in [0,1]. In this study, $\lambda$ is set to 0.7 \cite{maji2020systematic}.

\begin{algorithm}[htbp]

	\caption{\textsc{MddColor}($G_k$, $f', mdd\_set$)}\label{alo-mdd}
	\KwIn{An undirected graph $G_k = (V_k,E_k)$ and a coloring of $f'$}
	\KwOut{A coloring $f$ of $G_k$}
	
	$f\leftarrow \emptyset$;

    Sort $V_k$ according to core decomposition;

    \If{with probability $\alpha$}{
        $mdd\_set \leftarrow$ \textsc{MddSort}($G_k$)\;
        \ForEach{core layer vertex set $C \subseteq V_k$}{
            Sort $C$ ascending by $mdd\_set$;
        }
    }

	\ForEach{$v\in mdd\_Sort$ in reverse order}{
		
		$f(v)\leftarrow$ the minimum available color;// color $v$ with a minimum available color
		
		\If{$f(v)\geq color\_num(f')$}{
			
			\textsc{ReColor}($v$);
		}
		
		\If{$f(v)< color\_num(f')$} {
			
			$f \leftarrow f \cup f(v)$ ;	
		}\Else { $f \leftarrow f'$ and \Return {$f$}}
	}
	
	\Return{$f$}
	
\end{algorithm}

\begin{algorithm}[!h]
	\caption{$\textsc{MddSort}(G_k)$} \label{mdd_sort}

	\KwIn{An undirected graph $G_k = (V_k,E_k)$}
	\KwOut{A ordered vertex set $mdd\_set$}
	\While{$V \neq \emptyset$}{
		
		$S \leftarrow$ the set of vertices $v\in V$ with the minimum value of $d_{mdd\_Set}^m(v)$;
		
		$mdd\_set \leftarrow mdd\_set \cup S$;
		
		$V \leftarrow V\setminus S$;
	}
	
	\Return{$mdd\_set$}
	
\end{algorithm}

Based on the {\color{black}$k$-core} and mixed degree, this paper proposes a greedy heuristic $\textsc{MddColor()}$ for coloring graphs. As shown in Algorithm \ref{alo-mdd}, the algorithm has two parts:  ordering vertices {\color{black}(lines 2--6) and coloring vertices (lines 7--10)}. $\textsc{MddColor()}$ takes a graph $G_k$ and a coloring $f'$  as input, and outputs an improved coloring $f$ of $G_k$.  
{\color{black}
After initializing \( f \), the algorithm sorts the remaining vertices \( V_k \) based on core decomposition. Subsequently, with a probability of \( \alpha \), the vertices within each core layer are reordered according to the \( mdd\_set \) generated by \textsc{MddSort}() (see Algorithm \ref{mdd_sort}). The function \textsc{MddColor()} creates a vertex ordering for \( mdd\_set \) by iteratively selecting vertices with the minimum mixed degree associated with the current \( mdd\_set \).
During each iteration, the algorithm identifies the set \( S \) of vertices that possess the minimum value of \( d_{mdd\_Set}^m(v) \), updates \( mdd\_set \) by sequentially adding the vertices from \( S \), and removes these vertices from \( V \) (lines 2--4). It is noteworthy that the probabilistic reordering of vertices within each core layer is designed to enhance the diversity of coloring solutions.
} 

{\color{black}Observe that when constructing $mdd\_set$, its size increases with iteration}. Therefore, the mixed degree of a vertex $v (\notin S)$ that is adjacent to a vertex in $S$ but not in $S$ becomes smaller in the following iterations. This allows for such a vertex $v$ to be removed as early as possible. This observation may support the hypothesis that the later a vertex is deleted, the more challenging it becomes to assign an appropriate color to that vertex. For this, {\color{black}in the subsequent procedure of Algorithm~\ref{alo-mdd}}, the vertices are greedily colored in reverse order of $mdd\_set$ (lines 7--14), i.e., repeatedly selecting a vertex at the end of $mdd\_set$ and assigning a minimum available color to the vertex (line 8). Note that if the color assigned to a new vertex exceeds the number of colors used by an input coloring $f'$, then we attempt to avoid adding a new color by a recolor procedure $\textsc{ReColor}()$ proposed by Rossi and Ahmed \cite{rossi2014coloring} (lines 9--10); in addition, if $\textsc{ReColor}()$  fails to reduce the color number, then $\textsc{MddColor()}$ is terminated and the input coloring $f'$ is returned (lines 11--14).

{\color{black} 
Utilizing a binary heap to maintain saturation degrees allows the time complexity of \textsc{Dsatur} to be reduced from \(O(n^2)\) to \(O((n + m) \log n)\)~\cite{lewis2021guide}, significantly enhancing its efficiency for sparse graphs. In contrast, \textsc{MddColor} exhibits a time complexity of \(O(m + n^2)\) (see Theorem~\ref{theorem:complexity}), but it tends to perform better when establishing an upper bound on the chromatic number. To strike a balance between efficiency and accuracy, the algorithm employs \textsc{MddColor} during the initial iteration or when the working graph is reduced, switching to \textsc{Dsatur} in other cases.
}

{\color{black}
\begin{theorem} \label{theorem:complexity}
\textsc{MddColor} (Algorithm~\ref{alo-mdd}) runs in $O(m+n^2)$ time.
\end{theorem}

\begin{proof}
First, $k$-core decomposition can be executed in \(O(n + m)\) time \cite{batagelj2003m}. Second, we utilize bucket sort to implement \textsc{MDDSort} (Algorithm~\ref{mdd_sort}), organizing vertices into buckets according to their \(d_{mdd\_Set}^m\) values. The algorithm involves iteratively selecting the smallest non-empty bucket, removing the vertices from it, and adding them to \(S\) while concurrently updating the \(d_{mdd\_Set}^m\) values for the neighboring vertices in \(V \setminus S\). This procedure can be completed in \(O(m + n)\) time. For sorting the vertices within each core layer, we can simply insert them into their designated buckets based on their index positions in \(mdd\_set\), which can be achieved in \(O(n)\) time. Finally, concerning the coloring phase, the algorithm must color \(n\) vertices, potentially involving recoloring operations, which leads to an \(O(n^2)\) complexity. In conclusion, the overall time complexity of \textsc{MddColor} is \(O(m + n^2)\).
\end{proof}
}


\section{Experimental Evaluation}\label{sec-experiment}

We experimentally compared the proposed \textsc{HyColor} algorithm with the state-of-the-art algorithms {\color{black}for GCP} \textsc{FastColor}, \textsc{LS+I-Dsatur}, and {\color{black}\textsc{GC-SLIM}}.

\subsection{Datasets}
{\color{black}
The dataset consists of four prominent benchmarks: DIMACS \cite{Johnson1996CliquesCA}, DIMACS10 \cite{sanders2014benchmarking}, the Stanford Large Network Dataset Collection (SNAP) \cite{snapnets}, and the Network Data Repository Benchmark \cite{nr}. In total, these benchmarks encompass 209 instances.

\begin{itemize}
    \item DIMACS\footnote{\url{http://cedric.cnam.fr/~porumbed/graphs/}}: The benchmark comprises instances derived from the DIMACS Implementation Challenge. It incorporates several small, dense graphs that are extensively utilized for evaluating the performance of algorithms related to minimum vertex cover~\cite{luo2019local}, graph coloring~\cite{hebrard2019hybrid}, and other associated problems.

    \item DIMACS10\footnote{\url{https://www.cc.gatech.edu/dimacs10}}: The instances in this benchmark are drawn from the 10th DIMACS Implementation Challenge, which was originally created to assess the effectiveness of graph partitioning and clustering algorithms. These instances feature large, sparse graphs that are frequently employed to evaluate the performance of algorithms addressing various $\mathcal{NP}$-hard problems, including the minimum vertex cover~\cite{luo2019local} and minimum dominating set~\cite{zhu2024dual}.

    \item SNAP\footnote{\url{https://snap.stanford.edu/data}}. This benchmark comprises a diverse set of real-world networks, including social networks, communication networks, and web graphs, which are extensively used to test algorithms for problems such as the minimum dominating set~\cite{zhu2024dual} and related challenges. These networks are abstractions of real-world entities; for instance, `soc-Slashdot0902' is derived from the Slashdot social network as of February 2009. By obtaining high-quality colorings for such instances, it becomes possible to identify closely connected user groups, which can be leveraged for recommendation systems and targeted marketing.

    \item Network Data Repository\footnote{\url{https://networkrepository.com}}: This benchmark represents the largest network repository, featuring thousands of datasets across more than 30 domains. It offers a wide variety of instances ranging from small to large, including biological networks, social networks, web networks, road networks, and collaborative networks.

\end{itemize}
}

\subsection{Experimental Settings} \label{sec:setting}

{\color{black}
\textsc{HyColor} is implemented in C++ and compiled with gcc 9.4.0 using the ``-O3'' optimization option. All experiments were conducted on a server featuring an Intel Xeon w7-3445 CPU operating at 2.60 GHz, along with 128 GB of RAM, running Ubuntu 20.04.5 LTS.
We conducted a comparison of \textsc{HyColor} against three state-of-the-art heuristic algorithms for GCP: \textsc{FastColor}\footnote{\url{http://lcs.ios.ac.cn/~caisw/Color.html}}~\cite{lin2017reduction}, \textsc{LS+I-Dsatur}\footnote{\url{https://bitbucket.org/gkatsi/gc-cdcl/}}~\cite{hebrard2019hybrid}, and \textsc{GC-SLIM}\footnote{\url{https://github.com/ASchidler/coloring/}}~\cite{schidler2023sat}. For each instance, each algorithm was executed ten times using seeds 1 through 10, with a cutoff time of 60 seconds. The parameters for \textsc{FastColor}, \textsc{LS+I-Dsatur}, and \textsc{GC-SLIM} were set according to their default values, and the parameter $\alpha$ in \textsc{HyColor} (Algorthm~\ref{alo-1}) was set to 0.2 according to Section~\ref{sec:p}.

Due to space constraints, we present the minimum and average number of colors identified by each algorithm over ten runs for each instance, formatted as `Min(Avg)', respectively. More comprehensive details, including instance properties and the average time each algorithm takes to reach the best solution, are provided in the supplementary materials (Tables SI--SVI) of this paper.
It is worth noting that \textsc{GC-SLIM} employs an adjacency matrix to store graph data, which leads to a significant increase in memory requirements as the number of vertices grows. As a result, under the 128 GB memory limitation, graphs with more than 280,000 vertices cannot be processed. This threshold aligns with the value stated in the original paper. 
Furthermore, as the algorithms only check the time limit at the end of each iteration, a single iteration may exceed the cutoff time for some large-scale instances, resulting in reported times that surpass the specified threshold.
}

{\color{black}
The experimental results are presented in Tables \ref{tab-DIMACS} through \ref{tab-NR2}. Within these tables, an asterisk (`$*$') signifies that the algorithm confirms the final `Min' value as the optimal chromatic number. The smallest values for `Min' and `Avg' are emphasized in bold.
}

\subsection{Results on the DIMACS Benchmark} \label{exp:dimacs}
{\color{black}

Table \ref{tab-DIMACS} outlines the experimental results obtained from the DIMACS benchmark. Out of the 18 instances evaluated, \textsc{HyColor} achieved the best coloring in 11 instances (over 61\%) compared to the four algorithms tested. Specifically, it outperformed \textsc{FastColor} in nine instances, surpassed \textsc{LS + I-Dsatur} in six instances, and matched \textsc{GC-SLIM} in 11 instances.
Remarkably, \textsc{HyColor} achieved a strictly better coloring than the other algorithms in two instances. It demonstrated equal performance with \textsc{LS + I-Dsatur}, while slightly trailing behind \textsc{GC-SLIM}, which identified a strictly superior coloring in four instances. In comparison, \textsc{FastColor} was only able to provide a strictly better coloring in one instance.
Additionally, \textsc{HyColor}, \textsc{FastColor}, and \textsc{LS + I-Dsatur} each established that the coloring solutions were optimal in one instance, achieving the chromatic number for that instance. In summary, \textsc{HyColor} outperforms both \textsc{FastColor} and \textsc{LS + I-Dsatur} on the DIMACS benchmark, although it underperformed slightly against \textsc{GC-SLIM}, suggesting that \textsc{HyColor} is effective on small, dense graphs.
}

\begin{table}[tbhp]
  \centering
  \caption{Results on DIMACS Benchmark}
  \resizebox{\linewidth}{!}{
    \begin{tabular}{lrrrrrrrrrr}
    \toprule
    \multicolumn{1}{l}{\multirow{2}[2]{*}{Instance}} & \multicolumn{1}{c}{\textsc{FastColor}} & \multicolumn{1}{c}{\textsc{LS+I-DSatur}} & \multicolumn{1}{c}{\textsc{GC-LSIM}} & \multicolumn{1}{c}{\textsc{HyColor}} \\
    \cmidrule(lr{0.5em}){2-2} \cmidrule(lr{0.5em}){3-3} \cmidrule(lr{0.5em}){4-4} \cmidrule(lr{0.5em}){5-5}
    & Min(Avg) & Min(Avg)  & Min(Avg) & Min(Avg) \\
    \midrule
    C2000.5 & 206(206) & 205(206.4) & \textbf{182}(\textbf{182.6}) & 204(205.4) \\
    C2000.9 & 44(44.8) & 44(44.7) & \textbf{39}(\textbf{39}) & 44(44.8) \\
    C4000.5 & 377(377) & 382(382) & \textbf{343}(\textbf{343}) & 370(375.2) \\
    DSJC1000.5 & 112(112) & 112(112.6) & \textbf{99}(\textbf{99}) & 112(112.3) \\
    gen200\_p0.9\_44 & \textbf{7}(\textbf{7}) & 8(8) & \textbf{7}(\textbf{7}) & \textbf{7}(\textbf{7}) \\
    gen200\_p0.9\_55 & \textbf{7}(\textbf{7}) & \textbf{7}(\textbf{7}) & \textbf{7}(\textbf{7}) & \textbf{7}(\textbf{7}) \\
    gen400\_p0.9\_75 & \textbf{11}(11.1) & 12(12) & \textbf{11}(\textbf{11}) & \textbf{11}(11.4) \\
    hamming10-4 & 39(39.7) & \textbf{32}(\textbf{32}) & 34(34) & 39(39.7) \\
    hamming8-4 & \textbf{16}*(\textbf{16}) & \textbf{16}(\textbf{16}) & \textbf{16}(\textbf{16}) & \textbf{16}*(\textbf{16}) \\
    MANN\_a27 & 4(4) & \textbf{3}*(\textbf{3}) & 4(4) & 4(4) \\
    MANN\_a45 & \textbf{4}(\textbf{4}) & \textbf{4}(\textbf{4}) & \textbf{4}(\textbf{4}) & \textbf{4}(\textbf{4}) \\
    MANN\_a81 & \textbf{4}(\textbf{4}) & \textbf{4}(\textbf{4}) & \textbf{4}(\textbf{4}) & \textbf{4}(\textbf{4}) \\
    p\_hat1500-1 & 261(262.6) & 275(275) & \textbf{260}(\textbf{260}) & \textbf{260}(262) \\
    p\_hat1500-2 & 152(152.8) & 164(164) & 155(155) & \textbf{151}(\textbf{152.4}) \\
    p\_hat300-1 & \textbf{67}(\textbf{67.5}) & 68(68.6) & 68(68) & 68(68) \\
    p\_hat300-2 & \textbf{41}(41.9) & 42(42) & 43(43) & \textbf{41}(\textbf{41.6}) \\
    p\_hat700-1 & 137(137.8) & 151(151) & 138(138) & \textbf{136}(\textbf{137.3}) \\
    p\_hat700-2 & \textbf{81}(81.9) & 83(83.1) & 84(84) & \textbf{81}(\textbf{81.8}) \\
    \bottomrule
    \end{tabular}%
    }
  \label{tab-DIMACS}%
\end{table}%

\subsection{Results on the DIMACS10 Benchmark}
{\color{black}
Table~\ref{tab-DIMACS10} details the experimental results obtained from the DIMACS10 benchmark. Among the four algorithms assessed across all 36 instances, \textsc{HyColor} achieved the best coloring in 30 instances, representing over 83\% of the instances. It outperformed \textsc{FastColor} in 29 instances, surpassed \textsc{LS + I-Dsatur} in 24 instances, and exceeded \textsc{GC-SLIM} in 11 instances. Notably, \textsc{HyColor} identified a strictly superior coloring compared to the other algorithms in four instances, whereas \textsc{FastColor} achieved this in three instances, \textsc{LS + I-Dsatur} did not found strictly better colorings, and \textsc{GC-SLIM} achieved this also in three instances. Additionally, \textsc{HyColor}, \textsc{FastColor}, and \textsc{LS + I-Dsatur} demonstrated optimal coloring solutions in 23, 20, and 6 instances, respectively, successfully reaching the chromatic numbers for those instances.
In summary, \textsc{HyColor} outperforms the other algorithms in the DIMACS10 benchmark, exhibiting a superior ability to deliver better colorings in these sparse instances while also attaining optimal colorings in a greater number of instances.
}

\begin{table}[tbp]
  \centering
  \caption{Results on DIMACS10 Benchmark}
  \resizebox{\linewidth}{!}{
    \begin{tabular}{lrrrrrrrrrr}
    \toprule
   \multicolumn{1}{l}{\multirow{2}[2]{*}{Instance}} & \multicolumn{1}{c}{\textsc{FastColor}} & \multicolumn{1}{c}{\textsc{LS+I-DSatur}} & \multicolumn{1}{c}{\textsc{GC-LSIM}} & \multicolumn{1}{c}{\textsc{HyColor}} \\
    \cmidrule(lr{0.5em}){2-2} \cmidrule(lr{0.5em}){3-3} \cmidrule(lr{0.5em}){4-4} \cmidrule(lr{0.5em}){5-5}
    & Min(Avg) & Min(Avg)  & Min(Avg) & Min(Avg) \\
    \midrule
    as-22july06 & \textbf{17}*(\textbf{17}) & \textbf{17}*(\textbf{17}) & \textbf{17}(\textbf{17}) & \textbf{17}*(\textbf{17}) \\
    333SP & \textbf{5}(\textbf{5}) & \textbf{5}*(\textbf{5}) & N/A & \textbf{5}(\textbf{5}) \\
    audikw\_1 & 45(45) & 44(44) & N/A & \textbf{42}(\textbf{42}) \\
    belgium\_osm & \textbf{3}*(\textbf{3}) & \textbf{3}(\textbf{3}) & N/A & \textbf{3}*(\textbf{3}) \\
    cage15 & \textbf{14}(\textbf{14}) & 16(16) & N/A & \textbf{14}(\textbf{14}) \\
    caidaRouterLevel & \textbf{17}*(17.8) & \textbf{17}*(17.1) & \textbf{17}(\textbf{17}) & \textbf{17}*(\textbf{17}) \\
    citationCiteseer & \textbf{13}*(\textbf{13}) & \textbf{13}(\textbf{13}) & \textbf{13}(\textbf{13}) & \textbf{13}*(\textbf{13}) \\
    cnr-2000 & \textbf{84}*(\textbf{84}) & \textbf{84}(\textbf{84}) & N/A & \textbf{84}*(\textbf{84}) \\
    coAuthorsCiteseer & \textbf{87}*(\textbf{87}) & \textbf{87}(\textbf{87}) & \textbf{87}(\textbf{87}) & \textbf{87}*(\textbf{87}) \\
    coAuthorsDBLP & \textbf{115}*(\textbf{115}) & \textbf{115}(\textbf{115}) & N/A & \textbf{115}*(\textbf{115}) \\
    cond-mat-2005 & \textbf{30}*(\textbf{30}) & \textbf{30}(\textbf{30}) & \textbf{30}(\textbf{30}) & \textbf{30}*(\textbf{30}) \\
    coPapersCiteseer & \textbf{845}*(\textbf{845}) & \textbf{845}(\textbf{845}) & N/A & \textbf{845}*(\textbf{845}) \\
    coPapersDBLP & \textbf{337}*(\textbf{337}) & \textbf{337}(\textbf{337}) & N/A & \textbf{337}*(\textbf{337}) \\
    ecology1 & \textbf{2}*(\textbf{2}) & \textbf{2}(\textbf{2}) & N/A & \textbf{2}*(\textbf{2}) \\
    eu-2005 & \textbf{387}*(\textbf{387}) & \textbf{387}(\textbf{387}) & N/A & \textbf{387}*(\textbf{387}) \\
    G\_n\_pin\_pout & 6(6) & 6*(6) & \textbf{5}(\textbf{5}) & 6(6) \\
    in-2004 & \textbf{489}*(\textbf{489}) & \textbf{489}(\textbf{489}) & N/A & \textbf{489}*(\textbf{489}) \\
    kron\_g500-logn16 & 158(159.4) & 168(168) & 155(155) & \textbf{153}(\textbf{153.6}) \\
    kron\_g500-logn17 & \textbf{192}(193.3) & 208(208) & 195(195) & 193(\textbf{193.2}) \\
    kron\_g500-logn18 & 241(243.5) & 257(257) & 244(244.2) & \textbf{240}(\textbf{242.9}) \\
    kron\_g500-logn19 & \textbf{301}(\textbf{301}) & 319(319) & N/A & 302(302.4) \\
    kron\_g500-logn20 & 375(\textbf{375}) & 394(394) & N/A & \textbf{374}(376.3) \\
    kron\_g500-logn21 & \textbf{462}(\textbf{466.2}) & 485(485) & N/A & 464(466.6) \\
    ldoor & \textbf{33}(\textbf{33}) & 35(35) & N/A & \textbf{33}(\textbf{33}) \\
    luxembourg\_osm & \textbf{3}*(\textbf{3}) & \textbf{3}(\textbf{3}) & \textbf{3}(\textbf{3}) & \textbf{3}*(\textbf{3}) \\
    preferentialAttachment & \textbf{6}(\textbf{6}) & \textbf{6}*(\textbf{6}) & \textbf{6}(\textbf{6}) & \textbf{6}*(\textbf{6}) \\
    rgg\_n\_2\_17\_s0 & \textbf{15}*(\textbf{15}) & \textbf{15}(\textbf{15}) & \textbf{15}(\textbf{15}) & \textbf{15}*(\textbf{15}) \\
    rgg\_n\_2\_19\_s0 & \textbf{18}*(\textbf{18}) & \textbf{18}(\textbf{18}) & N/A & \textbf{18}*(\textbf{18}) \\
    rgg\_n\_2\_20\_s0 & \textbf{17}*(\textbf{17}) & \textbf{17}(\textbf{17}) & N/A & \textbf{17}*(\textbf{17}) \\
    rgg\_n\_2\_21\_s0 & \textbf{19}(\textbf{19}) & \textbf{19}(\textbf{19}) & N/A & \textbf{19}*(\textbf{19}) \\
    rgg\_n\_2\_22\_s0 & \textbf{20}*(\textbf{20}) & \textbf{20}(\textbf{20}) & N/A & \textbf{20}*(\textbf{20}) \\
    rgg\_n\_2\_23\_s0 & \textbf{21}*(\textbf{21}) & \textbf{21}(\textbf{21}) & N/A & \textbf{21}*(\textbf{21}) \\
    rgg\_n\_2\_24\_s0 & \textbf{21}(\textbf{21}) & \textbf{21}(\textbf{21}) & N/A & \textbf{21}*(\textbf{21}) \\
    smallworld & 7(7) & 8*(8) & \textbf{6}(\textbf{6}) & 7(7) \\
    uk-2002 & \textbf{944}*(\textbf{944}) & \textbf{944}(\textbf{944}) & N/A & \textbf{944}*(\textbf{944}) \\
    wave & 8(8) & 9(9) & \textbf{7}(\textbf{7}) & 8(8) \\
    \bottomrule
    \end{tabular}%
    }
  \label{tab-DIMACS10}%
\end{table}%

\subsection{Results on the SNAP Benchmark}
{\color{black}

Table \ref{tab-SNAP} presents the experimental results from the SNAP benchmark. Among the 22 instances analyzed, \textsc{HyColor} achieved the best coloring performance across all of them (100\%), surpassing the other three algorithms. In particular, \textsc{HyColor} outperformed \textsc{FastColor} in 17 instances, excelled over \textsc{LS + I-Dsatur} in 15 instances, and surpassed \textsc{GC-SLIM} in eight instances. Notably, \textsc{HyColor} identified a strictly better coloring than the other algorithms in five instances, whereas none of the other methods achieved a strictly better solution in any instance. 
Furthermore, \textsc{HyColor}, along with \textsc{FastColor} and \textsc{LS + I-Dsatur}, provided optimal coloring solutions in 11, 10, and one instance, respectively, achieving the chromatic number for those instances. Overall, \textsc{HyColor} distinctly outperforms the other algorithms on the SNAP benchmark, showcasing its ability to consistently find superior colorings and attain optimal colorings in more instances, particularly excelling with sparse graphs.
}

\begin{table}[tbhp]
  \centering
  \caption{Results on SNAP Benchmark}
  \resizebox{\linewidth}{!}{
    \begin{tabular}{lrrrrrrrrrr}
    \toprule
   \multicolumn{1}{l}{\multirow{2}[2]{*}{Instance}} & \multicolumn{1}{c}{\textsc{FastColor}} & \multicolumn{1}{c}{\textsc{LS+I-DSatur}} & \multicolumn{1}{c}{\textsc{GC-LSIM}} & \multicolumn{1}{c}{\textsc{HyColor}} \\
    \cmidrule(lr{0.5em}){2-2} \cmidrule(lr{0.5em}){3-3} \cmidrule(lr{0.5em}){4-4} \cmidrule(lr{0.5em}){5-5}
    & Min(Avg) & Min(Avg)  & Min(Avg) & Min(Avg) \\
    \midrule
    amazon0302 & \textbf{7}*(\textbf{7}) & \textbf{7}(\textbf{7}) & \textbf{7}(\textbf{7}) & \textbf{7}*(\textbf{7}) \\
    amazon0312 & \textbf{11}*(\textbf{11}) & \textbf{11}(\textbf{11}) & N/A & \textbf{11}*(\textbf{11}) \\
    amazon0505 & \textbf{11}*(\textbf{11}) & \textbf{11}(\textbf{11}) & N/A & \textbf{11}*(\textbf{11}) \\
    amazon0601 & \textbf{11}*(\textbf{11}) & \textbf{11}(\textbf{11}) & N/A & \textbf{11}*(\textbf{11}) \\
    cit-HepPh & \textbf{411}*(\textbf{411}) & \textbf{411}(\textbf{411}) & \textbf{411}(\textbf{411}) & \textbf{411}*(\textbf{411}) \\
    cit-HepTh & \textbf{562}*(\textbf{562}) & \textbf{562}(\textbf{562}) & \textbf{562}(\textbf{562}) & \textbf{562}*(\textbf{562}) \\
    cit-Patents & 12(12) & 12(12) & N/A & \textbf{11}*(\textbf{11}) \\
    email-EuAll & \textbf{18}(18.7) & 19(19.8) & 21(21) & \textbf{18}(\textbf{18.6}) \\
    p2p-Gnutella04 & \textbf{5}(\textbf{5}) & \textbf{5}(\textbf{5}) & \textbf{5}(\textbf{5}) & \textbf{5}(\textbf{5}) \\
    p2p-Gnutella24 & \textbf{5}(\textbf{5}) & \textbf{5}(\textbf{5}) & \textbf{5}(\textbf{5}) & \textbf{5}(\textbf{5}) \\
    p2p-Gnutella25 & \textbf{5}(\textbf{5}) & \textbf{5}(\textbf{5}) & \textbf{5}(\textbf{5}) & \textbf{5}(\textbf{5}) \\
    p2p-Gnutella30 & \textbf{5}(\textbf{5}) & \textbf{5}(\textbf{5}) & \textbf{5}(\textbf{5}) & \textbf{5}(\textbf{5}) \\
    p2p-Gnutella31 & \textbf{5}(\textbf{5}) & \textbf{5}(\textbf{5}) & \textbf{5}(\textbf{5}) & \textbf{5}(\textbf{5}) \\
    soc-Epinions1 & 30(30.7) & 30(30) & 30(30) & \textbf{28}(\textbf{28.5}) \\
    soc-Slashdot0811 & 30(30.5) & 30(30) & 30(30) & \textbf{29}(\textbf{29}) \\
    soc-Slashdot0902 & 31(31.6) & 31(31) & 31(31) & \textbf{29}(\textbf{29.9}) \\
    web-BerkStan & \textbf{201}*(\textbf{201}) & \textbf{201}(\textbf{201}) & N/A & \textbf{201}*(\textbf{201}) \\
    web-Google & \textbf{44}*(\textbf{44}) & \textbf{44}*(\textbf{44}) & N/A & \textbf{44}*(\textbf{44}) \\
    web-NotreDame & \textbf{155}*(\textbf{155}) & \textbf{155}(\textbf{155}) & N/A & \textbf{155}*(\textbf{155}) \\
    web-Stanford & \textbf{61}*(\textbf{61}) & \textbf{61}(\textbf{61}) & N/A & \textbf{61}*(\textbf{61}) \\
    wiki-Talk & 52(52) & 56(56) & N/A & \textbf{48}(\textbf{48.9}) \\
    Wiki-Vote & \textbf{22}(\textbf{22}) & 23(23) & 23(23) & \textbf{22}(\textbf{22}) \\
    \bottomrule
    \end{tabular}%
    }
  \label{tab-SNAP}%
\end{table}%

\subsection{Results on the Network Data Repository Benchmark} \label{exp:nr}
{\color{black}

Tables \ref{tab-NR} and \ref{tab-NR2} display the experimental results from the Network Data Repository benchmark. Among the 133 instances in this benchmark, \textsc{HyColor} achieved the best coloring in 131 of them, which accounts for over 98\%. This performance surpassed the other three algorithms, as \textsc{HyColor} outperformed \textsc{FastColor} in 99 instances, \textsc{LS + I-Dsatur} in 91 instances, and \textsc{GC-SLIM} in 77 instances. Notably, \textsc{HyColor} found a strictly superior coloring compared to the other algorithms in 23 instances, while \textsc{FastColor}, \textsc{LS + I-Dsatur}, and \textsc{GC-SLIM} achieved strictly better colorings in zero, one, and one instances, respectively.
Additionally, \textsc{HyColor}, \textsc{FastColor}, and \textsc{LS + I-Dsatur} demonstrated optimal colorings in 93, 83, and 15 instances, respectively, achieving the chromatic number for each of those instances.
In summary, \textsc{HyColor} outperforms the other algorithms in the Network Data Repository benchmark, showcasing a remarkable ability to identify better colorings and attain optimal colorings across a greater number of instances, particularly excelling in the context of sparse graphs.

}

\begin{table}[htbp]
  \centering
  \caption{Results on Network Data Repository Benchmark (I)}
  \resizebox{\linewidth}{!}{
    \begin{tabular}{lrrrrrrrrrr}
    \toprule
    \multicolumn{1}{l}{\multirow{2}[2]{*}{Instance}} & \multicolumn{1}{c}{\textsc{FastColor}} & \multicolumn{1}{c}{\textsc{LS+I-DSatur}} & \multicolumn{1}{c}{\textsc{GC-LSIM}} & \multicolumn{1}{c}{\textsc{HyColor}} \\
    \cmidrule(lr{0.5em}){2-2} \cmidrule(lr{0.5em}){3-3} \cmidrule(lr{0.5em}){4-4} \cmidrule(lr{0.5em}){5-5}
    & Min(Avg) & Min(Avg)  & Min(Avg) & Min(Avg) \\
    \midrule
    bio-celegans & \textbf{9}*(\textbf{9}) & \textbf{9}*(\textbf{9}) & \textbf{9}(\textbf{9}) & \textbf{9}*(\textbf{9}) \\
    bio-diseasome & \textbf{11}*(\textbf{11}) & \textbf{11}(\textbf{11}) & \textbf{11}(\textbf{11}) & \textbf{11}*(\textbf{11}) \\
    bio-dmela & \textbf{7}*(\textbf{7}) & \textbf{7}*(\textbf{7}) & \textbf{7}(\textbf{7}) & \textbf{7}*(\textbf{7}) \\
    bio-yeast & \textbf{6}*(\textbf{6}) & \textbf{6}(\textbf{6}) & \textbf{6}(\textbf{6}) & \textbf{6}*(\textbf{6}) \\
    bn-human*890*on\_1 & 99(107.1) & 113(113) & N/A & \textbf{96}(\textbf{98.7}) \\
    bn-human*914*on\_2 & 284(\textbf{284}) & 309(309) & N/A & \textbf{283}(284.2) \\
    bn-human*916*on\_1 & 256(256) & 277(277) & N/A & \textbf{250}(\textbf{252.3}) \\
    bn-human*110650 & 241(241) & 251(251) & N/A & \textbf{239}(\textbf{240.6}) \\
    bn-human*125334 & 164(164) & 175(175) & N/A & \textbf{162}(\textbf{163.9}) \\
    ca-AstroPh & \textbf{57}*(\textbf{57}) & \textbf{57}(\textbf{57}) & \textbf{57}(\textbf{57}) & \textbf{57}*(\textbf{57}) \\
    ca-citeseer & \textbf{87}*(\textbf{87}) & \textbf{87}(\textbf{87}) & \textbf{87}(\textbf{87}) & \textbf{87}*(\textbf{87}) \\
    ca-coauthors-dblp & \textbf{337}*(\textbf{337}) & \textbf{337}(\textbf{337}) & N/A & \textbf{337}*(\textbf{337}) \\
    ca-CondMat & \textbf{26}*(\textbf{26}) & \textbf{26}(\textbf{26}) & \textbf{26}(\textbf{26}) & \textbf{26}*(\textbf{26}) \\
    ca-CSphd & \textbf{3}*(\textbf{3}) & \textbf{3}(\textbf{3}) & \textbf{3}(\textbf{3}) & \textbf{3}*(\textbf{3}) \\
    ca-dblp-2010 & \textbf{75}*(\textbf{75}) & \textbf{75}(\textbf{75}) & \textbf{75}(\textbf{75}) & \textbf{75}*(\textbf{75}) \\
    ca-dblp-2012 & \textbf{114}*(\textbf{114}) & \textbf{114}(\textbf{114}) & N/A & \textbf{114}*(\textbf{114}) \\
    ca-Erdos992 & \textbf{8}*(\textbf{8}) & \textbf{8}(\textbf{8}) & \textbf{8}(\textbf{8}) & \textbf{8}*(\textbf{8}) \\
    ca-GrQc & \textbf{44}*(\textbf{44}) & \textbf{44}(\textbf{44}) & \textbf{44}(\textbf{44}) & \textbf{44}*(\textbf{44}) \\
    ca-HepPh & \textbf{239}*(\textbf{239}) & \textbf{239}(\textbf{239}) & \textbf{239}(\textbf{239}) & \textbf{239}*(\textbf{239}) \\
    ca-hollywood-2009 & \textbf{2209}*(\textbf{2209}) & \textbf{2209}*(\textbf{2209}) & N/A & \textbf{2209}*(\textbf{2209}) \\
    ca-MathSciNet & \textbf{25}*(\textbf{25}) & \textbf{25}(\textbf{25}) & N/A & \textbf{25}*(\textbf{25}) \\
    ca-netscience & \textbf{9}*(\textbf{9}) & \textbf{9}(\textbf{9}) & \textbf{9}(\textbf{9}) & \textbf{9}*(\textbf{9}) \\
    ia-email-EU & \textbf{13}(\textbf{13}) & \textbf{13}(\textbf{13}) & 14(14) & \textbf{13}(\textbf{13}) \\
    ia-email-univ & \textbf{12}*(\textbf{12}) & \textbf{12}(\textbf{12}) & \textbf{12}(\textbf{12}) & \textbf{12}*(\textbf{12}) \\
    ia-enron-large & \textbf{23}(\textbf{23.8}) & 24(24.6) & 25(25) & \textbf{23}(\textbf{23.8}) \\
    ia-enron-only & \textbf{8}*(\textbf{8}) & \textbf{8}(\textbf{8}) & \textbf{8}(\textbf{8}) & \textbf{8}*(\textbf{8}) \\
    ia-fb-messages & \textbf{6}(\textbf{6.1}) & 7(7) & 7(7) & \textbf{6}(\textbf{6.1}) \\
    ia-infect-dublin & \textbf{16}*(\textbf{16}) & \textbf{16}(\textbf{16}) & \textbf{16}(\textbf{16}) & \textbf{16}*(\textbf{16}) \\
    ia-infect-hyper & \textbf{16}(\textbf{16}) & \textbf{16}(\textbf{16}) & \textbf{16}(\textbf{16}) & \textbf{16}(\textbf{16}) \\
    ia-reality & \textbf{5}*(\textbf{5}) & \textbf{5}*(\textbf{5}) & \textbf{5}(\textbf{5}) & \textbf{5}*(\textbf{5}) \\
    ia-wiki-Talk & 26(26) & 26(26) & 26(26) & \textbf{25}(\textbf{25}) \\
    inf-power & \textbf{6}*(\textbf{6}) & \textbf{6}(\textbf{6}) & \textbf{6}(\textbf{6}) & \textbf{6}*(\textbf{6}) \\
    inf-roadNet-CA & \textbf{4}*(\textbf{4}) & \textbf{4}(\textbf{4}) & N/A & \textbf{4}*(\textbf{4}) \\
    inf-roadNet-PA & \textbf{4}*(\textbf{4}) & \textbf{4}(\textbf{4}) & N/A & \textbf{4}*(\textbf{4}) \\
    inf-road-usa & \textbf{4}*(\textbf{4}) & \textbf{4}(\textbf{4}) & N/A & \textbf{4}*(\textbf{4}) \\
    rec-amazon & \textbf{5}*(\textbf{5}) & \textbf{5}(\textbf{5}) & \textbf{5}(\textbf{5}) & \textbf{5}*(\textbf{5}) \\
    rt-retweet & \textbf{4}*(\textbf{4}) & \textbf{4}(\textbf{4}) & \textbf{4}(\textbf{4}) & \textbf{4}*(\textbf{4}) \\
    rt-retweet-crawl & \textbf{13}*(\textbf{13}) & \textbf{13}(\textbf{13}) & N/A & \textbf{13}*(\textbf{13}) \\
    rt-twitter-copen & \textbf{4}*(\textbf{4}) & \textbf{4}*(\textbf{4}) & \textbf{4}(\textbf{4}) & \textbf{4}*(\textbf{4}) \\
    scc\_enron-only & \textbf{120}*(\textbf{120}) & \textbf{120}(\textbf{120}) & \textbf{120}(\textbf{120}) & \textbf{120}*(\textbf{120}) \\
    scc\_fb-forum & \textbf{266}*(\textbf{266}) & \textbf{266}(\textbf{266}) & \textbf{266}(\textbf{266}) & \textbf{266}*(\textbf{266}) \\
    scc\_fb-messages & \textbf{707}*(\textbf{707}) & \textbf{707}(\textbf{707}) & \textbf{707}(\textbf{707}) & \textbf{707}*(\textbf{707}) \\
    scc\_infect-dublin & \textbf{84}*(\textbf{84}) & \textbf{84}(\textbf{84}) & \textbf{84}(\textbf{84}) & \textbf{84}*(\textbf{84}) \\
    scc\_infect-hyper & \textbf{106}*(\textbf{106}) & \textbf{106}(\textbf{106}) & \textbf{106}(\textbf{106}) & \textbf{106}*(\textbf{106}) \\
    scc\_reality & \textbf{1236}*(\textbf{1236}) & \textbf{1236}*(\textbf{1236}) & \textbf{1236}(\textbf{1236}) & \textbf{1236}*(\textbf{1236}) \\
    scc\_retweet & \textbf{166}*(\textbf{166}) & 167(167) & \textbf{166}(\textbf{166}) & \textbf{166}*(\textbf{166}) \\
    scc\_retweet-crawl & \textbf{20}*(\textbf{20}) & \textbf{20}(\textbf{20}) & N/A & \textbf{20}*(\textbf{20}) \\
    scc\_rt\_alwefaq & \textbf{16}*(\textbf{16}) & \textbf{16}(\textbf{16}) & \textbf{16}(\textbf{16}) & \textbf{16}*(\textbf{16}) \\
    scc\_rt\_assad & \textbf{8}*(\textbf{8}) & \textbf{8}(\textbf{8}) & \textbf{8}(\textbf{8}) & \textbf{8}*(\textbf{8}) \\
    scc\_rt\_bahrain & \textbf{8}*(\textbf{8}) & \textbf{8}(\textbf{8}) & \textbf{8}(\textbf{8}) & \textbf{8}*(\textbf{8}) \\
    scc\_rt\_barackobama & \textbf{10}*(\textbf{10}) & \textbf{10}(\textbf{10}) & \textbf{10}(\textbf{10}) & \textbf{10}*(\textbf{10}) \\
    scc\_rt\_damascus & \textbf{5}*(\textbf{5}) & \textbf{5}(\textbf{5}) & \textbf{5}(\textbf{5}) & \textbf{5}*(\textbf{5}) \\
    scc\_rt\_dash & \textbf{6}*(\textbf{6}) & \textbf{6}(\textbf{6}) & \textbf{6}(\textbf{6}) & \textbf{6}*(\textbf{6}) \\
    scc\_rt\_gmanews & \textbf{22}*(\textbf{22}) & \textbf{22}(\textbf{22}) & \textbf{22}(\textbf{22}) & \textbf{22}*(\textbf{22}) \\
    scc\_rt\_gop & \textbf{2}*(\textbf{2}) & \textbf{2}(\textbf{2}) & \textbf{2}(\textbf{2}) & \textbf{2}*(\textbf{2}) \\
    scc\_rt\_http & \textbf{3}*(\textbf{3}) & \textbf{3}(\textbf{3}) & \textbf{3}(\textbf{3}) & \textbf{3}*(\textbf{3}) \\
    scc\_rt\_israel & \textbf{2}*(\textbf{2}) & \textbf{2}(\textbf{2}) & \textbf{2}(\textbf{2}) & \textbf{2}*(\textbf{2}) \\
    scc\_rt\_justinbieber & \textbf{17}*(\textbf{17}) & \textbf{17}(\textbf{17}) & \textbf{17}(\textbf{17}) & \textbf{17}*(\textbf{17}) \\
    scc\_rt\_ksa & \textbf{6}*(\textbf{6}) & \textbf{6}(\textbf{6}) & \textbf{6}(\textbf{6}) & \textbf{6}*(\textbf{6}) \\
    scc\_rt\_lebanon & \textbf{2}*(\textbf{2}) & \textbf{2}(\textbf{2}) & \textbf{2}(\textbf{2}) & \textbf{2}*(\textbf{2}) \\
    scc\_rt\_libya & \textbf{3}*(\textbf{3}) & \textbf{3}(\textbf{3}) & \textbf{3}(\textbf{3}) & \textbf{3}*(\textbf{3}) \\
    scc\_rt\_lolgop & \textbf{42}*(\textbf{42}) & \textbf{42}(\textbf{42}) & \textbf{42}(\textbf{42}) & \textbf{42}*(\textbf{42}) \\
    scc\_rt\_mittromney & \textbf{5}*(\textbf{5}) & \textbf{5}(\textbf{5}) & \textbf{5}(\textbf{5}) & \textbf{5}*(\textbf{5}) \\
    scc\_twitter-copen & \textbf{581}*(\textbf{581}) & \textbf{581}(\textbf{581}) & \textbf{581}(\textbf{581}) & \textbf{581}*(\textbf{581}) \\
    sc-ldoor & \textbf{35}(\textbf{35}) & \textbf{35}(\textbf{35}) & N/A & \textbf{35}(\textbf{35}) \\
    sc-msdoor & 35(35) & \textbf{34}(\textbf{34.9}) & N/A & 35(35) \\
    soc-academia & \textbf{17}*(\textbf{17}) & \textbf{17}(\textbf{17}) & \textbf{17}(\textbf{17}) & \textbf{17}*(\textbf{17}) \\
    soc-BlogCatalog & 82(82) & 85(85) & 77(77) & \textbf{75}(\textbf{76}) \\
    soc-brightkite & \textbf{37}*(\textbf{37}) & \textbf{37}*(\textbf{37}) & \textbf{37}(\textbf{37}) & \textbf{37}*(\textbf{37}) \\
    soc-buzznet & 54(54) & 60(60) & 50(50) & \textbf{49}(\textbf{49.8}) \\
    soc-delicious & \textbf{21}*(\textbf{21}) & \textbf{21}(\textbf{21}) & N/A & \textbf{21}*(\textbf{21}) \\
    soc-digg & 63(63) & 73(73) & N/A & \textbf{55}(\textbf{55.4}) \\
    soc-dolphins & \textbf{5}*(\textbf{5}) & \textbf{5}(\textbf{5}) & \textbf{5}(\textbf{5}) & \textbf{5}*(\textbf{5}) \\
    socfb-A-anon & 28(28) & 34(34) & N/A & \textbf{26}(\textbf{26.2}) \\
    socfb-B-anon & 25(27.4) & 29(29) & N/A & \textbf{24}*(\textbf{24.3}) \\
    socfb-Berkeley13 & 43(43.4) & 44(44) & \textbf{42}(\textbf{42}) & \textbf{42}*(42.8) \\
    socfb-CMU & \textbf{45}*(45.5) & \textbf{45}*(\textbf{45}) & \textbf{45}(\textbf{45}) & \textbf{45}*(\textbf{45}) \\
    socfb-Duke14 & 43(43) & 45(45) & 46(46) & \textbf{40}(\textbf{40.3}) \\
    socfb-Indiana & 50(50) & 51(51) & \textbf{48}(\textbf{48}) & \textbf{48}*(48.5) \\
    socfb-MIT & 38(39.8) & 40(40) & 39(39) & \textbf{37}(\textbf{37.6}) \\
    socfb-OR & 32(32) & 34(34) & \textbf{31}(\textbf{31}) & \textbf{31}(\textbf{31}) \\
    socfb-Penn94 & 45(45) & 46(46) & \textbf{44}(\textbf{44}) & \textbf{44}*(\textbf{44}) \\
    socfb-Stanford3 & 55(55.2) & \textbf{51}*(\textbf{51}) & \textbf{51}(\textbf{51}) & \textbf{51}*(\textbf{51}) \\
    socfb-Texas84 & 54(54) & 57(57) & 55(55) & \textbf{52}(\textbf{52.9}) \\
    socfb-uci-uni & 8(8) & 8(8) & N/A & \textbf{7}(\textbf{7}) \\
    \bottomrule
    \end{tabular}%
    }
  \label{tab-NR}%
\end{table}%

\begin{table}[htbp]
  \centering
  \caption{Results on Network Data Repository Benchmark (II)}
  \resizebox{\linewidth}{!}{
    \begin{tabular}{lrrrrrrrrrr}
    \toprule
   \multicolumn{1}{l}{\multirow{2}[2]{*}{Instance}} & \multicolumn{1}{c}{\textsc{FastColor}} & \multicolumn{1}{c}{\textsc{LS+I-DSatur}} & \multicolumn{1}{c}{\textsc{GC-LSIM}} & \multicolumn{1}{c}{\textsc{HyColor}} \\
    \cmidrule(lr{0.5em}){2-2} \cmidrule(lr{0.5em}){3-3} \cmidrule(lr{0.5em}){4-4} \cmidrule(lr{0.5em}){5-5}
    & Min(Avg) & Min(Avg)  & Min(Avg) & Min(Avg) \\
    \midrule
    socfb-UCLA & \textbf{51}*(\textbf{51}) & \textbf{51}(\textbf{51}) & \textbf{51}(\textbf{51}) & \textbf{51}*(\textbf{51}) \\
    socfb-UConn & \textbf{50}(\textbf{50}) & \textbf{50}*(\textbf{50}) & \textbf{50}(\textbf{50}) & \textbf{50}*(\textbf{50}) \\
    socfb-UCSB37 & \textbf{53}(\textbf{53}) & \textbf{53}*(\textbf{53}) & 54(54) & \textbf{53}*(\textbf{53}) \\
    socfb-UF & 56(56.9) & 61(61) & \textbf{55}(\textbf{55}) & 56(56.9) \\
    socfb-UIllinois & \textbf{57}(\textbf{57}) & 58(58) & 58(58) & \textbf{57}*(\textbf{57}) \\
    socfb-Wisconsin87 & 39(39) & 41(41) & \textbf{38}(\textbf{38}) & \textbf{38}(38.4) \\
    soc-flickr & 101(101) & 108(108) & N/A & \textbf{93}(\textbf{94.4}) \\
    soc-flickr-und & 167(167.7) & 183(183) & N/A & \textbf{158}(\textbf{162}) \\
    soc-flixster & 36(36) & 39(39) & N/A & \textbf{35}(\textbf{35}) \\
    soc-FourSquare & \textbf{31}(31.3) & 34(34) & N/A & \textbf{31}(\textbf{31.1}) \\
    soc-google-plus & \textbf{66}*(\textbf{66}) & 67(67) & \textbf{66}(\textbf{66}) & \textbf{66}(\textbf{66}) \\
    soc-gowalla & \textbf{29}*(\textbf{29}) & \textbf{29}(\textbf{29}) & \textbf{29}(\textbf{29}) & \textbf{29}*(\textbf{29}) \\
    soc-karate & \textbf{5}*(\textbf{5}) & \textbf{5}(\textbf{5}) & \textbf{5}(\textbf{5}) & \textbf{5}*(\textbf{5}) \\
    soc-lastfm & 21(21) & 23(23) & N/A & \textbf{19}(\textbf{19.2}) \\
    soc-livejournal & \textbf{214}*(\textbf{214}) & \textbf{214}(\textbf{214}) & N/A & \textbf{214}*(\textbf{214}) \\
    soc-LiveJournal1 & \textbf{321}(\textbf{321}) & \textbf{321}(\textbf{321}) & N/A & \textbf{321}*(\textbf{321}) \\
    soc-LiveMocha & 33(33) & \textbf{26}(\textbf{26}) & 27(27) & \textbf{26}(\textbf{26}) \\
    soc-orkut & 71(\textbf{71}) & 77(77) & N/A & \textbf{68}(71.1) \\
    soc-pokec & \textbf{29}*(\textbf{29}) & \textbf{29}(\textbf{29}) & N/A & \textbf{29}*(\textbf{29}) \\
    soc-sign-Slashdot081106 & \textbf{29}(\textbf{29}) & 31(31) & 30(30) & \textbf{29}(\textbf{29}) \\
    soc-sign-Slashdot090221 & \textbf{29}(29.7) & 31(31) & 31(31) & \textbf{29}(\textbf{29.4}) \\
    soc-slashdot & 29(29.7) & 30(30) & 30(30) & \textbf{28}(\textbf{28.9}) \\
    soc-slashdot-trust-all & \textbf{29}(\textbf{29}) & 31(31) & 30(30) & \textbf{29}(\textbf{29}) \\
    soc-twitter-follows & \textbf{6}*(\textbf{6}) & 8(8) & N/A & \textbf{6}*(\textbf{6}) \\
    soc-twitter-follows-mun & \textbf{6}*(\textbf{6}) & 7(7) & N/A & \textbf{6}*(\textbf{6}) \\
    soc-twitter-higgs & \textbf{71}*(\textbf{71}) & 72(72) & N/A & \textbf{71}*(\textbf{71}) \\
    soc-wiki-Talk-dir & \textbf{49}(\textbf{49}) & 56(56) & N/A & \textbf{49}(\textbf{49}) \\
    soc-wiki-Vote & \textbf{7}*(\textbf{7}) & \textbf{7}*(\textbf{7}) & \textbf{7}(\textbf{7}) & \textbf{7}*(\textbf{7}) \\
    soc-youtube & 25(25) & 24(24) & N/A & \textbf{23}(\textbf{23}) \\
    soc-youtube-snap & 25(25) & 27(27) & N/A & \textbf{23}(\textbf{23.2}) \\
    tech-as-caida2007 & \textbf{16}*(\textbf{16}) & \textbf{16}(\textbf{16}) & \textbf{16}(\textbf{16}) & \textbf{16}*(\textbf{16}) \\
    tech-as-skitter & 68(68) & \textbf{67}(\textbf{67}) & N/A & \textbf{67}*(\textbf{67}) \\
    tech-internet-as & \textbf{16}*(\textbf{16}) & \textbf{16}(\textbf{16}) & \textbf{16}(\textbf{16}) & \textbf{16}*(\textbf{16}) \\
    tech-p2p-gnutella & \textbf{5}(\textbf{5}) & \textbf{5}(\textbf{5}) & \textbf{5}(\textbf{5}) & \textbf{5}(\textbf{5}) \\
    tech-RL-caida & 18(18) & \textbf{17}*(17.3) & \textbf{17}(\textbf{17}) & \textbf{17}*(\textbf{17}) \\
    tech-routers-rf & \textbf{16}*(\textbf{16}) & \textbf{16}(\textbf{16}) & \textbf{16}(\textbf{16}) & \textbf{16}*(\textbf{16}) \\
    tech-WHOIS & \textbf{58}*(\textbf{58}) & \textbf{58}*(\textbf{58}) & 59(59) & \textbf{58}*(\textbf{58}) \\
    web-arabic-2005 & \textbf{102}*(\textbf{102}) & \textbf{102}(\textbf{102}) & \textbf{102}(\textbf{102}) & \textbf{102}*(\textbf{102}) \\
    web-edu & \textbf{30}*(\textbf{30}) & \textbf{30}(\textbf{30}) & \textbf{30}(\textbf{30}) & \textbf{30}*(\textbf{30}) \\
    web-google & \textbf{18}*(\textbf{18}) & \textbf{18}(\textbf{18}) & \textbf{18}(\textbf{18}) & \textbf{18}*(\textbf{18}) \\
    web-indochina-2004 & \textbf{50}*(\textbf{50}) & \textbf{50}(\textbf{50}) & \textbf{50}(\textbf{50}) & \textbf{50}*(\textbf{50}) \\
    web-it-2004 & \textbf{432}*(\textbf{432}) & \textbf{432}(\textbf{432}) & N/A & \textbf{432}*(\textbf{432}) \\
    web-polblogs & \textbf{10}(\textbf{10}) & \textbf{10}(\textbf{10}) & \textbf{10}(\textbf{10}) & \textbf{10}(\textbf{10}) \\
    web-sk-2005 & \textbf{82}*(\textbf{82}) & \textbf{82}(\textbf{82}) & \textbf{82}(\textbf{82}) & \textbf{82}*(\textbf{82}) \\
    web-spam & \textbf{20}*(\textbf{20}) & \textbf{20}*(\textbf{20}) & \textbf{20}(\textbf{20}) & \textbf{20}*(\textbf{20}) \\
    web-uk-2005 & \textbf{500}*(\textbf{500}) & \textbf{500}(\textbf{500}) & \textbf{500}(\textbf{500}) & \textbf{500}*(\textbf{500}) \\
    web-webbase-2001 & \textbf{33}*(\textbf{33}) & \textbf{33}(\textbf{33}) & \textbf{33}(\textbf{33}) & \textbf{33}*(\textbf{33}) \\
    web-wikipedia2009 & \textbf{31}*(\textbf{31}) & \textbf{31}(\textbf{31}) & N/A & \textbf{31}*(\textbf{31}) \\
    \bottomrule
    \end{tabular}%
    }
  \label{tab-NR2}%
\end{table}%

\subsection{Summary of Results}
{\color{black}

A summary of the performance of \textsc{FastColor}, \textsc{LS+I-DSatur}, \textsc{GC-LSIM}, and \textsc{HyColor} across the DIMACS, DIMACS10, SNAP, and Network Data Repository benchmarks (NDR) is provided in Table~\ref{tab:summary}. 
The evaluation metrics employed are as follows: an asterisk (`*') indicates the number of instances where the algorithm achieves optimal coloring; a less-than symbol (`$<$') signifies instances in which the algorithm produces a coloring with fewer colors than its peers; and a less-than-or-equal-to symbol (`$<=$') denotes instances where the algorithm identifies the best coloring solution. It is noteworthy that \textsc{GC-LSIM} does not provide proof of optimality for its solutions, which is why there is no corresponding `*' column for it.
For more comprehensive details, please refer to Sections~\ref{exp:dimacs}--\ref{exp:nr}. Additionally, Figure~\ref{fig:summary} presents a more detailed visualization of the results, showcasing the superiority of \textsc{HyColor}. In this figure, the solid area represents the proportion of the best solutions identified by each algorithm across the respective benchmarks, while the dashed area indicates the proportion of strictly superior solutions found by each algorithm in those benchmarks.

\textsc{HyColor} consistently exhibited strong performance across all four benchmarks, with notable excellence in the DIMACS10, SNAP, and Network Data Repository datasets, which include real-world sparse graphs. It achieved both optimal and best coloring solutions, surpassing other algorithms in terms of strictly superior outcomes. \textsc{FastColor} also delivered commendable results, particularly in the Network Data Repository; however, it did not reach the same level as \textsc{HyColor} regarding strictly superior instances.
The algorithms \textsc{LS+I-DSatur}+ and \textsc{GC-LSIM} showcased competitive performance in certain benchmarks, but they lacked consistency in attaining optimal and strictly superior solutions throughout all tests. In summary, \textsc{HyColor} emerges as the most effective algorithm, successfully identifying both optimal and best solutions more frequently and consistently across a variety of graph instances.

}

\begin{table}[!t]
  \centering
  \caption{Summary Results}
  \resizebox{\linewidth}{!}{
    \begin{tabular}{lrrrrrrrrrrrr}
    \toprule
    \multicolumn{1}{l}{\multirow{2}[2]{*}{Benchmark}} & \multicolumn{3}{c}{\textsc{FastColor}} & \multicolumn{3}{c}{\textsc{LS+I-DSatur}} & \multicolumn{2}{c}{\textsc{GC-LSIM}} & \multicolumn{3}{c}{\textsc{HyColor}} \\
    \cmidrule(lr{0.1em}){2-4} \cmidrule(lr{0.1em}){5-7} \cmidrule(lr{0.1em}){8-9} \cmidrule(lr{0.1em}){10-12}
          & $<$ & $<=$  & $*$ & $<$ & $<=$  & $*$ & $<$ & $<=$ & $<$ & $<=$  & $*$ \\
    \midrule
    DIMACS & 1     & 9     & 1     & 2     & 6     & 1     & \textbf{4}     & \textbf{11}     & 2     & \textbf{11}    & \textbf{1}  \\
    DIMACS10 & 3     & 29    & 20    & 0     & 24    & 6     & 3     & 11       & \textbf{4}     & \textbf{30}    & \textbf{23}  \\
    SNAP  & 0     & 17    & 10    & 0     & 15    & 1     & 0     & 8      & \textbf{5}     & \textbf{22}    & \textbf{11}  \\
    NDR& 0     & 99    & 83    & 1     & 91    & 15    & 1     & 77      & \textbf{23}    & \textbf{131}   & \textbf{93}  \\
    \midrule
    Total & 4     & 154   & 114   & 3     & 136   & 23    & 8     & 107    & \textbf{34}    & \textbf{194}   & \textbf{128}  \\
    \bottomrule
    \end{tabular}%
    }
  \label{tab:summary}%
\end{table}%

\begin{figure}[htbp]
\centering
	\includegraphics[width=0.7\linewidth]{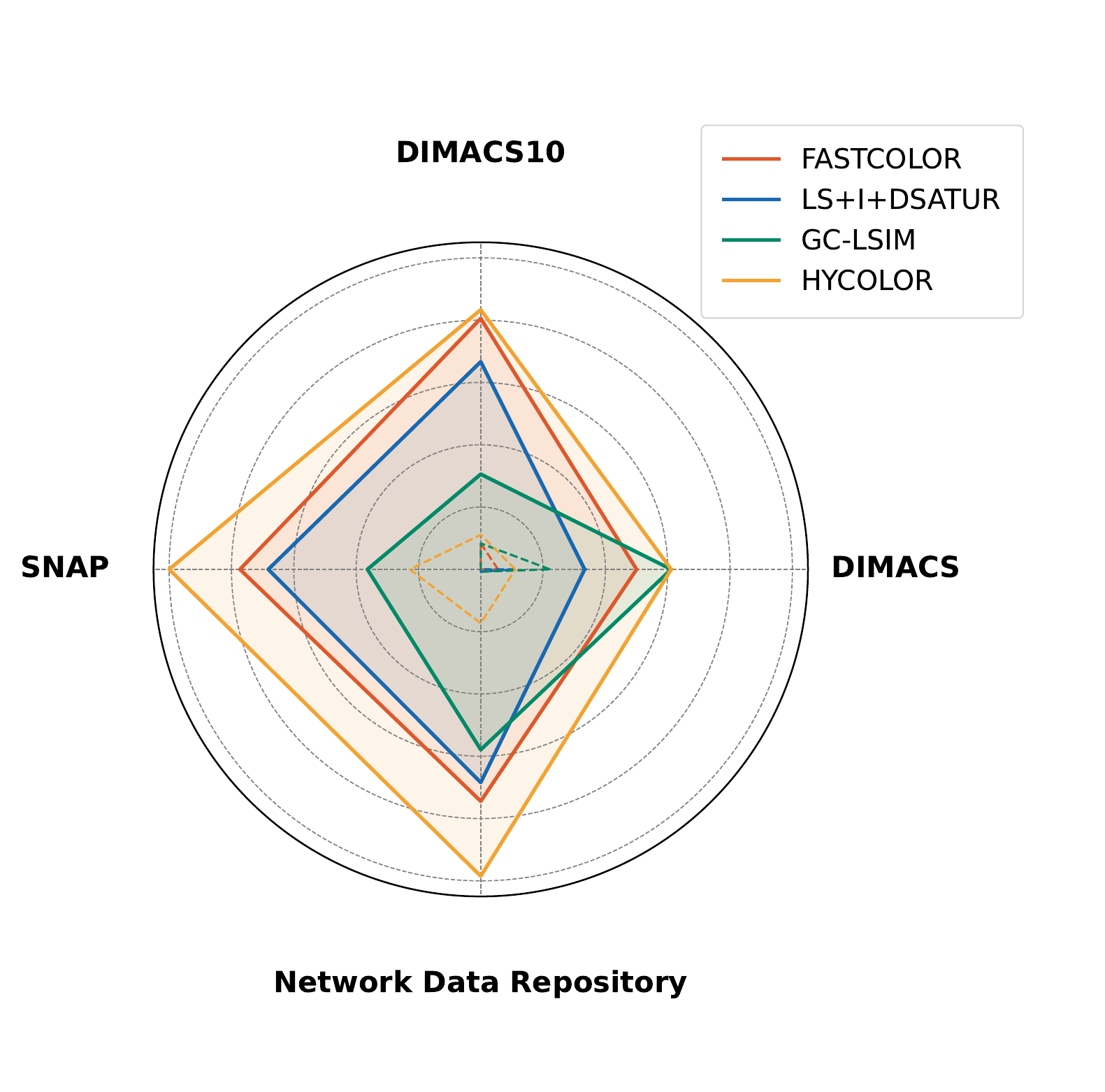}
	\caption{Visualization of Summary Results}
	\label{fig:summary}
\end{figure}

\subsection{Statistical Test}
{\color{black}

To evaluate the performance differences between our proposed algorithm, \textsc{HyColor}, and the baseline algorithms, we performed the Friedman \cite{friedman1937use} and Nemenyi \cite{demvsar2006statistical} tests.

The Friedman test is a non-parametric statistical method designed to detect differences in performance ranks among multiple algorithms across various datasets. We applied the Friedman test to ascertain whether \textsc{HyColor}’s performance significantly differs from the baseline algorithms. A low $p$-value (typically below 0.05) suggests significant algorithm differences. When the Friedman test identifies significant differences, a post-hoc test, such as the Nemenyi test, can be employed to determine which specific pairs of algorithms exhibit differences. The Nemenyi test utilizes the concept of a critical difference (CD) to assess whether the performance ranks of the two algorithms differ significantly. If the difference in their average ranks is less than the CD value, it indicates no significant difference between them; otherwise, a significant difference is present.
The CD is calculated as follows:

\begin{equation} \label{equ1}
CD = q_\alpha \cdot \sqrt{\frac{k (k + 1)}{6N}}
\end{equation}

In Equation \ref{equ1}, \( q_\alpha \) indicates the critical value derived from the Studentized Range Distribution Table. This value is determined by the significance level \( \alpha \) and the number of algorithms, represented by \( k \). Furthermore, \( N \) denotes the number of datasets or problem instances considered in the analysis. In our experiment, \( N \) is set at 209, and \( k \) equals 4. The corresponding critical value is \( q_\alpha = 2.569 \), which allows us to compute \( CD \approx 0.3244 \).

Based on the data presented in Tables~\ref{tab-DIMACS}–\ref{tab-NR2}, we ranked the algorithms by their average performance (`Avg') across ten runs for each instance and conducted a Friedman test. The results yielded a $p$-value of $6.4319 \times 10^{-28}$, indicating significant differences among the algorithms at a 0.05 significance level. As a result, we proceeded with the Nemenyi test for a more detailed comparison of the algorithms, with the findings illustrated in Fig.~\ref{fig:CD}. This figure shows that the differences between the algorithms exceed the CD value, confirming significant distinctions. Importantly, \textsc{HyColor} demonstrated superior performance compared to the baseline algorithms on average across 209 instances.
}

\begin{figure}[htbp]
\centering
	\includegraphics[width=\linewidth]{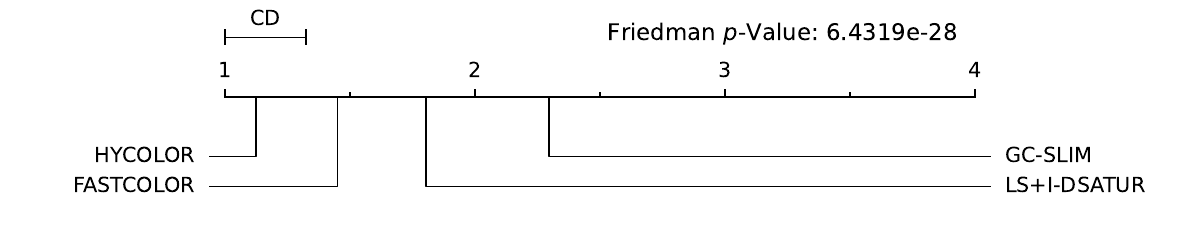}
	\caption{Statistic results with the Nemenyi test among four algorithms across 209 instances}
	\label{fig:CD}
\end{figure}

{\color{black}
\subsection{Effects of Different Parameter Settings} \label{sec:p}

As detailed in Section~\ref{sec-mddcolor}, the \textsc{MddColor} component of \textsc{HyColor} features a tunable parameter, \(\alpha\). To assess the influence of different parameter settings, we conducted an empirical evaluation of \textsc{HyColor} using \(\alpha\) values ranging from 0.1 to 1.0 in increments of 0.1. Our aim is for \textsc{HyColor} to minimize the number of colors used while demonstrating enhanced performance across instances from four benchmarks.
For each parameter setting, we ran \textsc{HyColor} ten times with a time limit of 60 seconds, recording the  `Min' and `Avg' ranks achieved on each benchmark. The average ranking was then used to evaluate overall performance. Figure~\ref{fig:ptest} illustrates the average performance of \textsc{HyColor} across the various parameter settings and benchmarks. Notably, it appears that \(\alpha = 0.2\) yielded the best performance, prompting us to configure \(\alpha\) for \textsc{HyColor} at this optimal value.
}

\begin{figure}[htbp]
\centering
	\includegraphics[width=0.9\linewidth]{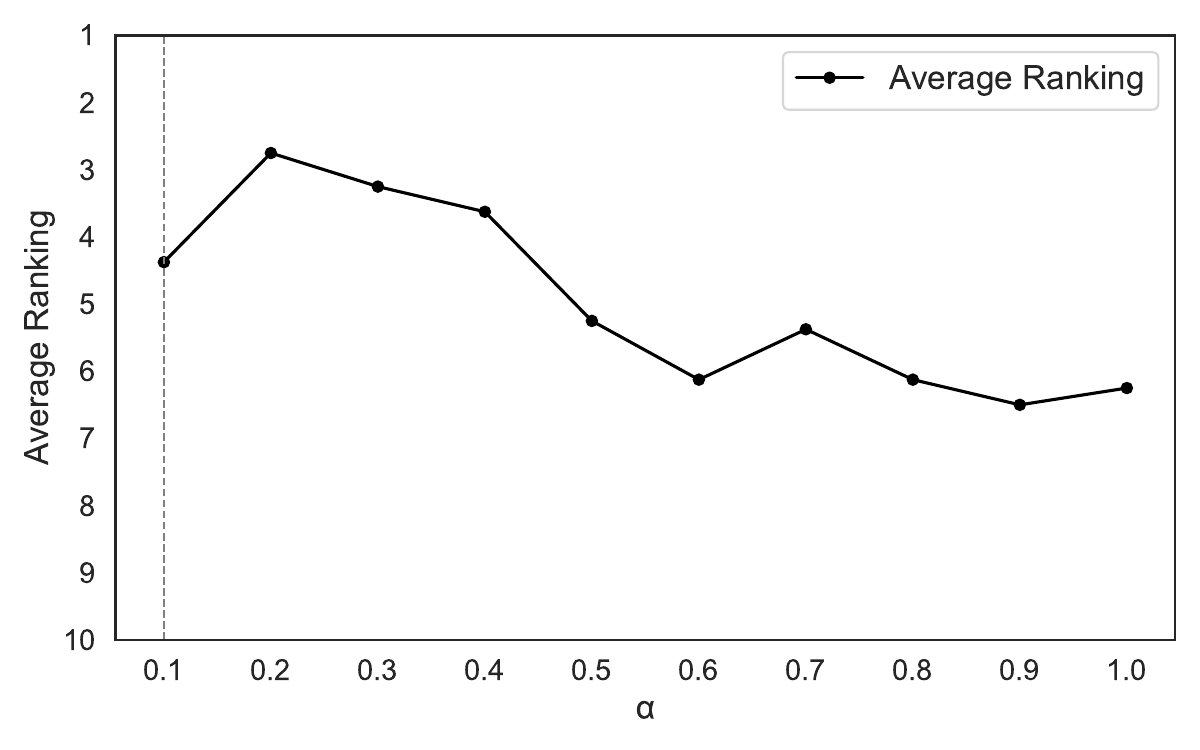}
	\caption{Effects of Different Parameter Settings}
	\label{fig:ptest}
\end{figure}

{\color{black}
\subsection{Effects of Main Components}

To assess the effectiveness of the proposed strategies, we performed tests on three variants of \textsc{HyColor}. These variants include: \textsc{HyColor} without the \textsc{ReduRule} component, referred to as `\textsc{without\_redu}'; \textsc{HyColor} without the MDD sorting component, known as `\textsc{without\_mdd}'; and \textsc{HyColor} without the \textsc{Exactlb()} function, designated as `\textsc{without\_exactlb}'. In the \textsc{HyColor} algorithm, the \textsc{ReduRule} and MDD-based sorting components play crucial roles in enhancing the quality of the coloring solution. In contrast, the \textsc{ExactLb()} function is designed to refine the clique identified by \textsc{FindClq(),} which serves as the lower bound for the chromatic number.
We compared the average `Min' values obtained from \textsc{without\_redu} and \textsc{without\_mdd} with the results achieved by the \textsc{FastColor}, \textsc{LS+I-DSatur}, and \textsc{HyColor} algorithms across a variety of benchmarks. Furthermore, we evaluated the lower bounds produced by \textsc{HyColor} in comparison to those from \textsc{without\_exactlb}. It is noteworthy that we did not include results from \textsc{GC-LSIM} in this analysis due to its inability to solve certain instances within the 128GB memory limit. As a result, the average `Min' values obtained by \textsc{GC-LSIM} on these benchmarks cannot be directly compared to those of the other algorithms.

In the same configuration outlined in Section~\ref{sec:setting}, the performance of each algorithm across the four benchmarks is presented in Figure~\ref{fig:comparable}. The horizontal axis illustrates the average `Min' values achieved by each algorithm on each benchmark, while the vertical axis enumerates the algorithm names. A longer bar signifies a higher average `Min' value, indicating poorer performance by the algorithm. Conversely, a shorter bar reflects a lower average `Min' value, suggesting improved performance. 
{\color{black}
As depicted in the figure, \textsc{HyColor} consistently outperforms the ablation algorithms across all benchmarks, except for a minor disadvantage on DIMACS when compared to \textsc{without\_mdd}. This occurs because MDD is generally more effective for analyzing sparse and complex networks, such as social networks. However, its applicability is limited, and it tends to exhibit higher time complexity when dealing with dense graphs.
}

Furthermore, we selected four representative graphs (a dense graph, two real-world graphs of different scales, and a large-scale graph used in the 10th DIMACS Implementation Challenge) to demonstrate the performance of \textsc{without\_exactlb} and \textsc{HyColor} in finding the chromatic number lower bound, as shown in Figure~\ref{fig:without_exactlb}. It can be observed that \textsc{HyColor} improves the lower bound more quickly and is sometimes able to find a better lower bound compared to \textsc{without\_exactlb} (e.g., in the instance of `soc-Livejournal1').
}

\begin{figure*}[htbp]
\centering
	\includegraphics[width=\linewidth]{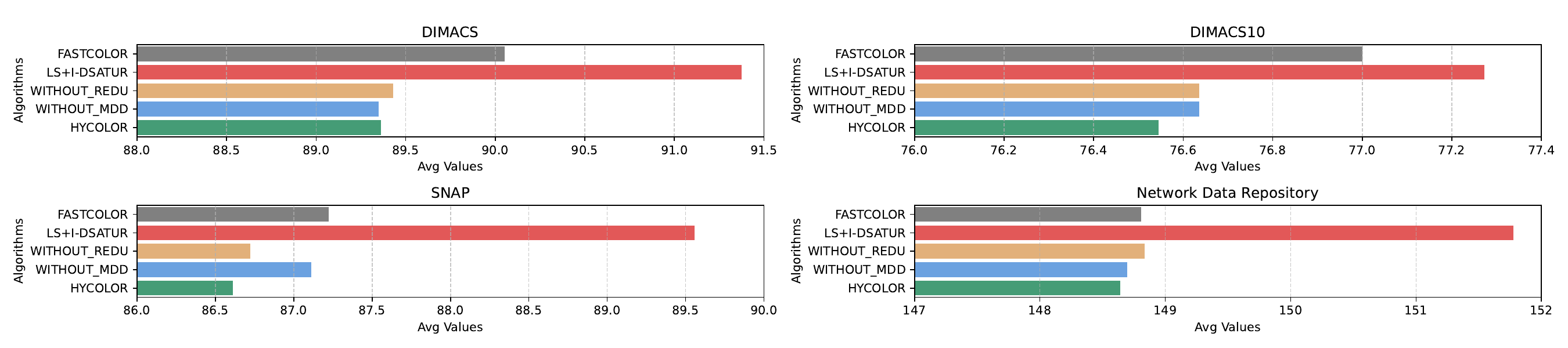}
	\caption{Effects of ReduRule and MDD Components}
	\label{fig:comparable}
\end{figure*}

\begin{figure*}[htbp]
\centering
	\includegraphics[width=\linewidth]{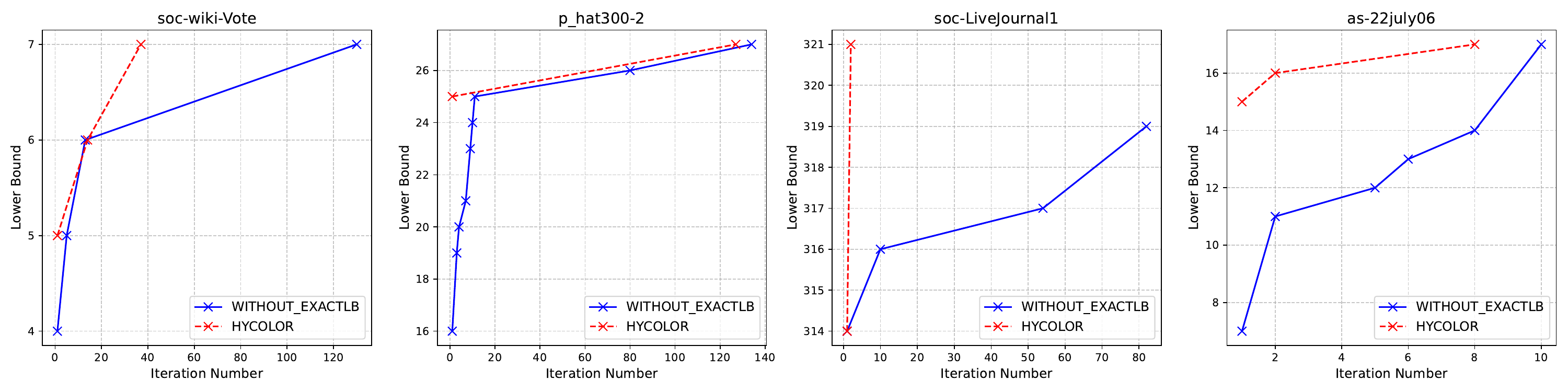}
	\caption{Effects of ExactLB Component}
	\label{fig:without_exactlb}
\end{figure*}

\section{Conclusions}\label{sec-conclusion}
{\color{black}
The performance of heuristic algorithms for coloring large-scale sparse graphs has seen little advancement over the past five years. In this paper, we introduce an algorithm called \textsc{HyColor} to address GCP. When compared with leading heuristic algorithms, such as \textsc{FastColor} and \textsc{LS + I-DSatur}, which are particularly effective on large-scale sparse graphs, \textsc{HyColor} demonstrates efficient coloring capabilities while providing solutions of equal or superior quality. Additionally, \textsc{HyColor} performs competitively on small dense graphs.

In contrast to \textsc{GC-SLIM}, an algorithm designed for dense graphs, \textsc{HyColor} offers significant improvements in large-scale sparse graph coloring while maintaining comparable performance on small dense graphs. Notably, \textsc{HyColor} surpasses \textsc{GC-SLIM} on certain specialized graphs, such as `p\_hat1500-2', `p\_hat300-2', `p\_hat700-1', and `p\_hat700-2', highlighting its versatility and adaptability across diverse graph types and sizes. While the experimental results indicate that \textsc{HyColor} effectively balances performance on both large-scale sparse and small dense graphs, it still shows limitations on some small dense graphs, including `C2000*', `C4000.*', and `DSJC*'. Improving \textsc{HyColor}’s performance on these dense graph instances remains a key objective for future research, as this would enhance its overall robustness and applicability in graph coloring tasks.

This study has focused exclusively on vertex coloring. However, there are other significant coloring problems, such as partition coloring and total coloring, which present additional challenges, especially in the context of large-scale graphs. Developing optimization algorithms to tackle these problems represents a considerable challenge, and the concepts and algorithmic framework introduced in this work may be extended to these broader areas. Exploring such adaptations for other coloring problems is a promising direction for future research.
}

\bibliographystyle{IEEEtran}

\bibliography{cas-refs}

\begin{thebibliography}{10}
\providecommand{\url}[1]{#1}
\csname url@samestyle\endcsname
\providecommand{\newblock}{\relax}
\providecommand{\bibinfo}[2]{#2}
\providecommand{\BIBentrySTDinterwordspacing}{\spaceskip=0pt\relax}
\providecommand{\BIBentryALTinterwordstretchfactor}{4}
\providecommand{\BIBentryALTinterwordspacing}{\spaceskip=\fontdimen2\font plus
\BIBentryALTinterwordstretchfactor\fontdimen3\font minus \fontdimen4\font\relax}
\providecommand{\BIBforeignlanguage}[2]{{%
\expandafter\ifx\csname l@#1\endcsname\relax
\typeout{** WARNING: IEEEtran.bst: No hyphenation pattern has been}%
\typeout{** loaded for the language `#1'. Using the pattern for}%
\typeout{** the default language instead.}%
\else
\language=\csname l@#1\endcsname
\fi
#2}}
\providecommand{\BIBdecl}{\relax}
\BIBdecl

\bibitem{Hou23}
C.~Hou, ``Optimization of smart sensor for balance between code bug ratio and energy consumption,'' \emph{{IEEE} Transactions on Systems, Man, and Cybernetcis-Systems}, vol.~53, no.~1, pp. 451--461, 2023.

\bibitem{TianSZTJ23}
Y.~Tian, L.~Si, X.~Zhang, K.~C. Tan, and Y.~Jin, ``Local model-based pareto front estimation for multiobjective optimization,'' \emph{{IEEE} Transactions on Systems, Man, and Cybernetcis-Systems}, vol.~53, no.~1, pp. 623--634, 2023.

\bibitem{WOS:001068977000001}
N.~Wu, Y.~Xu, H.~Wang, and D.~M. Kilgour, ``Matrix representation and behavioral analysis in a graph model for conflict resolution with incomplete fuzzy preferences,'' \emph{{IEEE} Transactions on Systems, Man, and Cybernetcis-Systems}, p. r 10.1109/TSMC.2023.3307362, 2023.

\bibitem{brelaz1979new}
D.~Brelaz, ``New methods to color the vertices of a graph. communication,'' 1979.

\bibitem{four-color}
O.~Ore, \emph{The Four-color problem}.\hskip 1em plus 0.5em minus 0.4em\relax Academic Press, New York, 1969.

\bibitem{hajos}
C.~P. A, ``Haj{\'{o}}s' graph-coloring conjecture: variations and counterexamples,'' \emph{Journal of Combinatorial Theory Series B}, vol.~26, pp. 268--274, 1979.

\bibitem{greene2002new}
J.~E. Greene, ``A new short proof of kneser's conjecture,'' \emph{The American mathematical monthly}, vol. 109, no.~10, pp. 918--920, 2002.

\bibitem{gondran2019optimality}
A.~Gondran and L.~Moalic, ``Optimality clue for graph coloring problem,'' in \emph{International Conference on Integration of Constraint Programming, Artificial Intelligence, and Operations Research}.\hskip 1em plus 0.5em minus 0.4em\relax Springer, 2019, pp. 337--354.

\bibitem{xue2022graph}
H.~Xue, V.~Rajan, and Y.~Lin, ``Graph coloring via neural networks for haplotype assembly and viral quasispecies reconstruction,'' \emph{Advances in Neural Information Processing Systems, NeurIPS 2022}, vol.~35, pp. 30\,898--30\,910, 2022.

\bibitem{WOS:000631202400013}
D.~Ma, X.~Hu, H.~Zhang, Q.~Sun, and X.~Xie, ``A hierarchical event detection method based on spectral theory of multidimensional matrix for power system,'' \emph{{IEEE} Transactions on Systems, Man, and Cybernetcis-Systems}, vol.~51, no.~4, pp. 2173--2186, 2021.

\bibitem{sadavare2012review}
A.~Sadavare and R.~Kulkarni, ``A review of application of graph theory for network,'' \emph{International Journal of Computer Science and Information Technologies}, vol.~3, no.~6, pp. 5296--5300, 2012.

\bibitem{gupta2022scalable}
S.~Gupta and S.~Amin, ``Scalable design of error-correcting output codes using discrete optimization with graph coloring,'' in \emph{Advances in Neural Information Processing Systems, NeurIPS 2022}, 2022.

\bibitem{9825671}
Y.~Li, X.~Li, and L.~Gao, ``An effective solution space clipping-based algorithm for large-scale permutation flow shop scheduling problem,'' \emph{{IEEE} Transactions on Systems, Man, and Cybernetcis-Systems}, vol.~53, no.~1, pp. 635--646, 2023.

\bibitem{zhu2020partition}
E.~Zhu, F.~Jiang, C.~Liu, and J.~Xu, ``Partition independent set and reduction-based approach for partition coloring problem,'' \emph{IEEE Transactions on Cybernetics}, vol.~52, no.~6, pp. 4960--4969, 2022.

\bibitem{garey1979}
M.~R. Garey and D.~S. Johnson, \emph{Computers and Intractability: A Guide to the Theory of NP-completeness}.\hskip 1em plus 0.5em minus 0.4em\relax Freeman, San Francisco, CA, USA, 1979.

\bibitem{almeida2016priority}
B.~F. Almeida, I.~Correia, and F.~Saldanha-da Gama, ``Priority-based heuristics for the multi-skill resource constrained project scheduling problem,'' \emph{Expert Systems with Applications}, vol.~57, pp. 91--103, 2016.

\bibitem{zuck2007}
D.~Zuckerman., ``Linear degree extractors and the inapproximability of max clique and chromatic number,'' \emph{Theory of Computing}, vol.~3, no.~1, pp. 103--128, 2007.

\bibitem{3-coloring}
R.~Beigel and D.~Eppstein, ``3-coloring in time $o (1.3289^n)$,'' \emph{Journal of Algorithms}, vol.~54, no.~2, pp. 168--204, 2005.

\bibitem{4-coloring}
F.~V. Fomin, S.~Gaspers, and S.~saurabh, ``Improved exact algorithms for counting 3-and 4-colorings,'' in \emph{International Computing and Combinatorics Conference}, 2007, pp. 65--74.

\bibitem{count3}
E.~Zhu, P.~Wu, and Z.~Shao, ``Exact algorithms for counting 3-colorings of graphs,'' \emph{Discrete Applied Mathematics}, vol. 322, pp. 74--93, 2022.

\bibitem{newman2003structure}
M.~E. Newman, ``The structure and function of complex networks,'' \emph{SIAM review}, vol.~45, no.~2, pp. 167--256, 2003.

\bibitem{shen2012new}
Y.~Shen, D.~T. Nguyen, Y.~Xuan, and M.~T. Thai, ``New techniques for approximating optimal substructure problems in power-law graphs,'' \emph{Theoretical Computer Science}, vol. 447, pp. 107--119, 2012.

\bibitem{porumbel2013informed}
D.~C. Porumbel, J.-K. Hao, and P.~Kuntz, ``Informed reactive tabu search for graph coloring,'' \emph{Asia-Pacific Journal of Operational Research}, vol.~30, no.~04, p. 1350010, 2013.

\bibitem{furini2017improved}
F.~Furini, V.~Gabrel, and I.-C. Ternier, ``An improved dsatur-based branch-and-bound algorithm for the vertex coloring problem,'' \emph{Networks}, vol.~69, no.~1, pp. 124--141, 2017.

\bibitem{cseker2021exact}
O.~{\c{S}}eker, T.~Ekim, and Z.~C. Ta{\c{s}}k{\i}n, ``An exact cutting plane algorithm to solve the selective graph coloring problem in perfect graphs,'' \emph{European Journal of Operational Research}, vol. 291, no.~1, pp. 67--83, 2021.

\bibitem{xu2022cuckoo}
Y.~Xu and Y.~Chen, ``A cuckoo quantum evolutionary algorithm for the graph coloring problem,'' in \emph{Bio-Inspired Computing: Theories and Applications: 16th International Conference, BIC-TA 2021, Taiyuan, China, December 17--19, 2021, Revised Selected Papers, Part I}.\hskip 1em plus 0.5em minus 0.4em\relax Springer, 2022, pp. 88--99.

\bibitem{marappan2022new}
R.~Marappan and S.~Bhaskaran, ``New evolutionary operators in coloring dimacs challenge benchmark graphs,'' \emph{International Journal of Information Technology}, vol.~14, no.~6, pp. 3039--3046, 2022.

\bibitem{dokeroglu2021memetic}
T.~Dokeroglu and E.~Sevinc, ``Memetic teaching--learning-based optimization algorithms for large graph coloring problems,'' \emph{Engineering Applications of Artificial Intelligence}, vol. 102, p. 104282, 2021.

\bibitem{schidler2023sat}
A.~Schidler and S.~Szeider, ``Sat-boosted tabu search for coloring massive graphs,'' \emph{ACM Journal of Experimental Algorithmics}, vol.~28, pp. 1--19, 2023.

\bibitem{hebrard2019hybrid}
E.~Hebrard and G.~Katsirelos, ``A hybrid approach for exact coloring of massive graphs,'' in \emph{International Conference on Integration of Constraint Programming, Artificial Intelligence, and Operations Research}.\hskip 1em plus 0.5em minus 0.4em\relax Springer, 2019, pp. 374--390.

\bibitem{rossi2014coloring}
R.~A. Rossi and N.~K. Ahmed, ``Coloring large complex networks,'' \emph{Social Network Analysis and Mining}, vol.~4, no.~1, pp. 1--37, 2014.

\bibitem{verma2015solving}
A.~Verma, A.~Buchanan, and S.~Butenko, ``Solving the maximum clique and vertex coloring problems on very large sparse networks,'' \emph{INFORMS Journal on computing}, vol.~27, no.~1, pp. 164--177, 2015.

\bibitem{lin2017reduction}
J.~Lin, S.~Cai, C.~Luo, and K.~Su, ``A reduction based method for coloring very large graphs,'' in \emph{IJCAI}, 2017, pp. 517--523.

\bibitem{hertz1987using}
A.~Hertz and D.~d. Werra, ``Using tabu search techniques for graph coloring,'' \emph{Computing}, vol.~39, no.~4, pp. 345--351, 1987.

\bibitem{avanthay2003variable}
C.~Avanthay, A.~Hertz, and N.~Zufferey, ``A variable neighborhood search for graph coloring,'' \emph{European Journal of Operational Research}, vol. 151, no.~2, pp. 379--388, 2003.

\bibitem{blochliger2008graph}
I.~Bl{\"o}chliger and N.~Zufferey, ``A graph coloring heuristic using partial solutions and a reactive tabu scheme,'' \emph{Computers \& Operations Research}, vol.~35, no.~3, pp. 960--975, 2008.

\bibitem{wu2012coloring}
Q.~Wu and J.-K. Hao, ``Coloring large graphs based on independent set extraction,'' \emph{Computers \& Operations Research}, vol.~39, no.~2, pp. 283--290, 2012.

\bibitem{lu2010memetic}
Z.~L{\"u} and J.-K. Hao, ``A memetic algorithm for graph coloring,'' \emph{European Journal of Operational Research}, vol. 203, no.~1, pp. 241--250, 2010.

\bibitem{marappan2018solution}
R.~Marappan and G.~Sethumadhavan, ``Solution to graph coloring using genetic and tabu search procedures,'' \emph{Arabian Journal for Science and Engineering}, vol.~43, no.~2, pp. 525--542, 2018.

\bibitem{moalic2018variations}
L.~Moalic and A.~Gondran, ``Variations on memetic algorithms for graph coloring problems,'' \emph{Journal of Heuristics}, vol.~24, no.~1, pp. 1--24, 2018.

\bibitem{goudet2022deep}
O.~Goudet, C.~Grelier, and J.-K. Hao, ``A deep learning guided memetic framework for graph coloring problems,'' \emph{Knowledge-Based Systems}, p. 109986, 2022.

\bibitem{lodha2019sat}
N.~Lodha, S.~Ordyniak, and S.~Szeider, ``A sat approach to branchwidth,'' \emph{ACM Transactions on Computational Logic (TOCL)}, vol.~20, no.~3, pp. 1--24, 2019.

\bibitem{reichl2023circuit}
F.-X. Reichl, F.~Slivovsky, and S.~Szeider, ``Circuit minimization with qbf-based exact synthesis,'' in \emph{Proceedings of the AAAI Conference on Artificial Intelligence}, vol.~37, no.~4, 2023, pp. 4087--4094.

\bibitem{audemard2009predicting}
G.~Audemard and L.~Simon, ``Predicting learnt clauses quality in modern sat solvers,'' in \emph{Twenty-first international joint conference on artificial intelligence}.\hskip 1em plus 0.5em minus 0.4em\relax Citeseer, 2009.

\bibitem{fleury2020cadical}
A.~Fleury and M.~Heisinger, ``Cadical, kissat, paracooba, plingeling and treengeling entering the sat competition 2020,'' \emph{SAT COMPETITION}, vol. 2020, p.~50, 2020.

\bibitem{zhao2024iterative}
F.~Zhao, C.~Zhuang, L.~Wang, and C.~Dong, ``An iterative greedy algorithm with $ q $-learning mechanism for the multiobjective distributed no-idle permutation flowshop scheduling,'' \emph{IEEE Transactions on Systems, Man, and Cybernetics: Systems}, 2024.

\bibitem{7938718}
Y.~Wang, X.~Li, R.~Ruiz, and S.~Sui, ``An iterated greedy heuristic for mixed no-wait flowshop problems,'' \emph{IEEE Transactions on Cybernetics}, vol.~48, no.~5, pp. 1553--1566, 2018.

\bibitem{pan2020knowledge}
Z.~Pan, D.~Lei, and L.~Wang, ``A knowledge-based two-population optimization algorithm for distributed energy-efficient parallel machines scheduling,'' \emph{IEEE transactions on cybernetics}, vol.~52, no.~6, pp. 5051--5063, 2020.

\bibitem{zhao2022hyperheuristic}
F.~Zhao, S.~Di, and L.~Wang, ``A hyperheuristic with q-learning for the multiobjective energy-efficient distributed blocking flow shop scheduling problem,'' \emph{IEEE Transactions on Cybernetics}, vol.~53, no.~5, pp. 3337--3350, 2022.

\bibitem{zhu2024dual}
E.~Zhu, Y.~Zhang, S.~Wang, D.~Strash, and C.~Liu, ``A dual-mode local search algorithm for solving the minimum dominating set problem,'' \emph{Knowledge-Based Systems}, vol. 298, p. 111950, 2024.

\bibitem{7835722}
A.~S. Azad, M.~Islam, and S.~Chakraborty, ``A heuristic initialized stochastic memetic algorithm for mdpvrp with interdependent depot operations,'' \emph{IEEE Transactions on Cybernetics}, vol.~47, no.~12, pp. 4302--4315, 2017.

\bibitem{mostafaie2020systematic}
T.~Mostafaie, F.~M. Khiyabani, and N.~J. Navimipour, ``A systematic study on meta-heuristic approaches for solving the graph coloring problem,'' \emph{Computers \& Operations Research}, vol. 120, p. 104850, 2020.

\bibitem{meraihi2019chaotic}
Y.~Meraihi, A.~Ramdane-Cherif, M.~Mahseur, and D.~Achelia, ``A chaotic binary salp swarm algorithm for solving the graph coloring problem,'' in \emph{Modelling and Implementation of Complex Systems: Proceedings of the 5th International Symposium, MISC 2018, December 16-18, 2018, Laghouat, Algeria 5}.\hskip 1em plus 0.5em minus 0.4em\relax Springer, 2019, pp. 106--118.

\bibitem{silva2020improved}
A.~F.~d. Silva, L.~G.~A. Rodriguez, and J.~F. Filho, ``The improved colourant algorithm: a hybrid algorithm for solving the graph colouring problem,'' \emph{International Journal of Bio-Inspired Computation}, vol.~16, no.~1, pp. 1--12, 2020.

\bibitem{wang2012hybrid}
S.~Wang and J.~Watada, ``A hybrid modified pso approach to var-based facility location problems with variable capacity in fuzzy random uncertainty,'' \emph{Information Sciences}, vol. 192, pp. 3--18, 2012.

\bibitem{marappan2021solving}
R.~Marappan and G.~Sethumadhavan, ``Solving graph coloring problem using divide and conquer-based turbulent particle swarm optimization,'' \emph{Arabian Journal for Science and Engineering}, pp. 1--18, 2021.

\bibitem{cai2016fast}
S.~Cai and J.~Lin, ``Fast solving maximum weight clique problem in massive graphs.'' in \emph{IJCAI}, 2016, pp. 568--574.

\bibitem{seidman1983network}
S.~B. Seidman, ``Network structure and minimum degree,'' \emph{Social networks}, vol.~5, no.~3, pp. 269--287, 1983.

\bibitem{zeng2013ranking}
A.~Zeng and C.-J. Zhang, ``Ranking spreaders by decomposing complex networks,'' \emph{Physics letters A}, vol. 377, no.~14, pp. 1031--1035, 2013.

\bibitem{batagelj2003m}
V.~Batagelj and M.~Zaversnik, ``An o (m) algorithm for cores decomposition of networks. corr,'' \emph{arXiv preprint cs.DS/0310049}, vol.~37, 2003.

\bibitem{maji2020systematic}
G.~Maji, S.~Mandal, and S.~Sen, ``A systematic survey on influential spreaders identification in complex networks with a focus on k-shell based techniques,'' \emph{Expert Systems with Applications}, vol. 161, p. 113681, 2020.

\bibitem{lewis2021guide}
R.~Lewis, \emph{Guide to graph colouring}.\hskip 1em plus 0.5em minus 0.4em\relax Springer, 2021.

\bibitem{Johnson1996CliquesCA}
\BIBentryALTinterwordspacing
D.~J. Johnson and M.~A. Trick, ``Cliques, coloring, and satisfiability: Second dimacs implementation challenge, workshop, october 11-13, 1993,'' 1996. [Online]. Available: \url{https://api.semanticscholar.org/CorpusID:118207836}
\BIBentrySTDinterwordspacing

\bibitem{sanders2014benchmarking}
P.~Sanders, C.~Schulz, and D.~Wagner, ``Benchmarking for graph clustering and partitioning,'' \emph{Encyclopedia of social network analysis and mining Springer}, 2014.

\bibitem{snapnets}
J.~Leskovec and A.~Krevl, ``{SNAP Datasets}: {Stanford} large network dataset collection,'' \url{http://snap.stanford.edu/data}, Jun. 2014.

\bibitem{nr}
R.~A. Rossi and N.~K. Ahmed, ``The network data repository with interactive graph analytics and visualization,'' in \emph{AAAI}, 2015.

\bibitem{luo2019local}
C.~Luo, H.~H. Hoos, S.~Cai, Q.~Lin, H.~Zhang, and D.~Zhang, ``Local search with efficient automatic configuration for minimum vertex cover.'' in \emph{IJCAI}, 2019, pp. 1297--1304.

\bibitem{friedman1937use}
M.~Friedman, ``The use of ranks to avoid the assumption of normality implicit in the analysis of variance,'' \emph{Journal of the american statistical association}, vol.~32, no. 200, pp. 675--701, 1937.

\bibitem{demvsar2006statistical}
J.~Dem{\v{s}}ar, ``Statistical comparisons of classifiers over multiple data sets,'' \emph{The Journal of Machine learning research}, vol.~7, pp. 1--30, 2006.

\end{thebibliography}

\end{document}